\documentclass[envcountsame]{llncs}
\usepackage{amsfonts}
\usepackage{amssymb}
\usepackage{mathrsfs}
\usepackage[]{graphicx}
\usepackage{comment}
\usepackage{extarrows}
\usepackage{wrapfig}

\begin{document}
\excludecomment{cver}
\includecomment{fver}
%\excludecomment{fver}
%\includecomment{cver}
\title{Parameterized complexity results\\for $1$-safe Petri nets}
\author{M.~Praveen \and Kamal Lodaya}
%\date{}
\institute{The Institute of Mathematical Sciences, Chennai 600113, India}
\maketitle
%Natural numbers
\newcommand{\macNat}{\mathbb{N}}

%big O notation
\newcommand{\Oh}{\mathcal{O}}

%treewidth and pathwidth
\newcommand{\mactw}{\text{tw}}
\newcommand{\macpw}{p{\scriptstyle w}}
\newcommand{\macheight}{h}
\newcommand{\mactwo}{t{\scriptstyle w}}

%fonts
\newcommand{\maccc}{\textsc}
\newcommand{\name}{\textsf}
\newcommand{\concept}{\textbf}

%complexity classes
\newcommand{\macpspace}{\maccc{Pspace}}
\newcommand{\macexpsp}{\maccc{Expspace}}
\newcommand{\macparapsp}{\maccc{ParaPspace}}
\newcommand{\macfpt}{\maccc{Fpt}}
\newcommand{\macwone}{\maccc{W[1]}}
\newcommand{\macwtwo}{\maccc{W[2]}}
\newcommand{\macxp}{\maccc{Xp}}
\newcommand{\macnp}{\maccc{Np}}
\newcommand{\macparanp}{\maccc{ParaNp}}
\newcommand{\macconp}{\maccc{Co-Np}}
\newcommand{\macp}{\maccc{Ptime}}
\newcommand{\macyes}{\maccc{Yes}}
\newcommand{\macoh}{\mathcal{O}}
\newcommand{\macilp}{\maccc{Ilp}}
\newcommand{\macsdp}{\maccc{Signed Digraph Pebbling}}
\newcommand{\macpspl}{\maccc{Propositional STRIPS Planning}}

\newcommand{\macppwcnfsat}{p-\textsc{Pw-Sat}}
\newcommand{\macnlcp}{\textsc{Nlcp}}
\newcommand{\maccsp}{\textsc{Csp}}

\newcommand{\macpopterm}[1]{\name{#1}}

\newcommand{\macpset}[1]{\mathscr{P}(#1)}

\newcommand{\macrestr}[1]{\lceil {#1}}

%propositional variables
\newcommand{\macpv}{\Phi}
\newcommand{\macpvo}{q}
\newcommand{\macpvh}{s}
\newcommand{\macpvf}{t}
\newcommand{\macpvidx}{i}
\newcommand{\macclause}{C}
\newcommand{\macnumcl}{m}
\newcommand{\macclidx}{j}

%Partitioned weighted satisfiability
\newcommand{\macpcnff}{\mathcal{F}}
\newcommand{\macpart}{part}
\newcommand{\mactget}{tg}
\newcommand{\macpartidx}{r}
\newcommand{\macnumpv}{n}
\newcommand{\macnumpart}{k}

%For referring to definitions, theorems etc.
\newcommand{\Defref}[1]{Definition~\ref{#1}}
\newcommand{\defref}[1]{Def.~\ref{#1}}
\newcommand{\Thmref}[1]{Theorem~\ref{#1}}
\newcommand{\thmref}[1]{Theorem~\ref{#1}}
\newcommand{\corref}[1]{Corollary~\ref{#1}}
\newcommand{\Lemref}[1]{Lemma~\ref{#1}}
\newcommand{\lemref}[1]{Lemma~\ref{#1}}
\newcommand{\Propref}[1]{Proposition~\ref{#1}}
\newcommand{\propref}[1]{Prop.~\ref{#1}}
\newcommand{\Claimref}[1]{Claim~\ref{#1}}
\newcommand{\claimref}[1]{Claim~\ref{#1}}
\newcommand{\algoref}[1]{Algorithm~\ref{#1}}
\newcommand{\Algoref}[1]{Algorithm~\ref{#1}}
\newcommand{\Figref}[1]{Figure~\ref{#1}}
\newcommand{\figref}[1]{Fig.~\ref{#1}}
\newcommand{\Tabref}[1]{Table~\ref{#1}}
\newcommand{\tabref}[1]{Table~\ref{#1}}
\newcommand{\Secref}[1]{Section~\ref{#1}}
\newcommand{\secref}[1]{Sect.~\ref{#1}}
\newcommand{\Chref}[1]{Chapter~\ref{#1}}
\newcommand{\chref}[1]{Chap.~\ref{#1}}
\newcommand{\apndref}[1]{Appendix~\ref{#1}}
\newcommand{\Apndref}[1]{App.~\ref{#1}}

%Graphs
\newcommand{\macgraph}{G}
\newcommand{\macvertexo}{v}
\newcommand{\macverto}{v}
\newcommand{\macvertext}{s}
\newcommand{\macuedge}{e}
\newcommand{\macedge}{e}
\newcommand{\macvertexs}{V}
\newcommand{\macedges}{E}
\newcommand{\maccols}{S}
\newcommand{\maccolor}{\ell}
\newcommand{\macatleast}{atLeast}
\newcommand{\macatmost}{atMost}
\newcommand{\macproper}{proper}
\newcommand{\macbag}{B}
\newcommand{\macvertcov}{VC}
\newcommand{\macvcnum}{k}
\newcommand{\macpath}{\mu}
\newcommand{\macdgraph}{D}
\newcommand{\macarcs}{A}
\newcommand{\macred}{Red}
\newcommand{\macblue}{Blue}
\newcommand{\macbvertex}{b}
\newcommand{\mactree}{\mathcal{T}}
\newcommand{\mactn}{\tau}
\newcommand{\macnodes}{\mathit{nodes}}

%Definition
\newcommand{\macdef}{\overset{\bigtriangleup}{=}}

%Manufacturing example
\newcommand{\macfovo}{x}
\newcommand{\macfovt}{y}
\newcommand{\macrmo}{\alpha}
\newcommand{\macrmt}{\beta}
\newcommand{\macrmh}{\gamma}

%Petri net for simulating partitioned weighted sat
\newcommand{\macserpl}{s}
\newcommand{\macglpl}{g}
\newcommand{\macttr}{t}
\newcommand{\macftr}{f}
\newcommand{\mactdel}{td}
\newcommand{\macfdel}{fd}
\newcommand{\mactcnt}{tc}
\newcommand{\macfcnt}{fc}
\newcommand{\macpartindo}{t{\scriptstyle \uparrow}}
\newcommand{\macpartindt}{f{\scriptstyle \uparrow}}
\newcommand{\macpartctro}{t{\scriptstyle u}}
\newcommand{\macpartctrt}{f{\scriptstyle l}}
\newcommand{\macdpthctro}{d}
\newcommand{\mactpl}{x}
\newcommand{\macfpl}{\overline{\mactpl}}
\newcommand{\mactcpl}{y}
\newcommand{\macfcpl}{\overline{\mactcpl}}
\newcommand{\macclpl}{C}
\newcommand{\macpartinds}{P_{1}}

%Petri nets
\newcommand{\macplace}{p}
\newcommand{\macplaceo}{p}
\newcommand{\mactrans}{t}
\newcommand{\mactranso}{t}
\newcommand{\macplaces}{P}
\newcommand{\macnumplaces}{m}
\newcommand{\mactranss}{T}
\newcommand{\macben}{ben}
\newcommand{\macpre}{\mathit{Pre}}
\newcommand{\macpost}{\mathit{Post}}
\newcommand{\macmark}{M}
\newcommand{\macmarks}{\pi}
\newcommand{\macnet}{\mathcal{N}}
\newcommand{\macplaceidx}{i}
\newcommand{\macplaceidxt}{j}
\newcommand{\macnumtranstype}{l}
\newcommand{\macvar}{\mathrm{\mathit{int}}}
\newcommand{\mactypeidx}{j}
\newcommand{\macarcw}{w}
\newcommand{\macplvar}{I}
\newcommand{\macplvars}{\mathit{Int}}
\newcommand{\macsplaces}{\mathcal{S}}
\newcommand{\macStep}[1]{\xLongrightarrow{#1}}
\newcommand{\macfirseqo}{\rho}
\newcommand{\macfslidx}{r}
\newcommand{\macipls}[1]{^{\bullet}#1}
\newcommand{\macopls}[1]{#1^{\bullet}}

%Constraint satisfaction problem
\newcommand{\macdom}{dom}
\newcommand{\macdeg}{deg}
\newcommand{\macdomidx}{d}
\newcommand{\macdegidx}{l}
\newcommand{\macvaro}{q}
\newcommand{\macvaridx}{i}
\newcommand{\macconstridx}{j}
\newcommand{\macnumvar}{n}
\newcommand{\macnumconstr}{m}
\newcommand{\macconstro}{C}
\newcommand{\macsbr}[1]{[#1]}

%LTL formulas
\newcommand{\macltlfo}{\phi}
\newcommand{\macnext}{X}
\newcommand{\macuntil}{U}
\newcommand{\macseqidx}{j}
\newcommand{\macseqidxt}{k}
\newcommand{\macseqidxh}{l}
\newcommand{\macfsa}{\mathcal{A}}
\newcommand{\macfsba}{\mathcal{B}}

%Automata
\newcommand{\macstates}{Q}
\newcommand{\macalph}{\Sigma}
\newcommand{\macletter}{\sigma}
\newcommand{\mactransrel}{\delta}
\newcommand{\macedgeconstr}{u}
\newcommand{\macstateo}{q}
\newcommand{\macfinss}{F}

%problem instances
\newcommand{\macprino}{I}
\newcommand{\macpro}{\Pi}
\newcommand{\macparamo}{\kappa}
\newcommand{\macalpho}{\Sigma}
\newcommand{\macparamdepo}{f}
\newcommand{\macparamdept}{g}
\newcommand{\macredmap}{A}
\newcommand{\macnuminst}{m}
\newcommand{\macinstidx}{i}

\begin{abstract}
  We associate a graph with a $1$-safe Petri net and study the
  parameterized complexity of various problems with parameters derived
  from the graph. With treewidth as the parameter, we give
  \macwone{}-hardness results for many problems about $1$-safe Petri
  nets. As a corollary, this proves a conjecture of Downey et.~al.
  about the hardness of some graph pebbling problems.  We consider the
  parameter benefit depth (that is known to be helpful in getting
  better algorithms for general Petri nets) and again give
  \macwone{}-hardness results for various problems on $1$-safe Petri
  nets.  We also consider the stronger parameter vertex cover number.
  Combining the well known automata-theoretic method and a powerful
  fixed parameter tractability (\macfpt{}) result about Integer Linear
  Programming, we give a \macfpt{} algorithm for model checking
  Monadic Second Order (MSO) formulas on $1$-safe Petri nets, with
  parameters vertex cover number and the size of the formula. 
\end{abstract}

\section{Introduction}
Petri nets are popular for modelling because they offer a succinct 
representation of loosely coupled communicating systems. 
Some powerful techniques are available but the complexity of analysis 
is high. In his lucid survey \cite{ESPCompl98}, Esparza summarizes 
the situation as follows: almost every interesting analysis question
on the behaviour of general Petri nets is \macexpsp-hard, and almost
every interesting analysis question on the behaviour of 1-safe Petri
nets is \macpspace-hard.  By considering special subclasses of nets
slightly better results can be obtained. 
%but even for 1-safe \emph{acyclic} nets reachability is \macnp-hard. 
Esparza points out that T-systems (also called marked graphs) and 
S-systems (essentially sequential transition systems)
are the largest subclasses where polynomial time algorithms are
available. We therefore look for a \emph{structural parameter} 
with respect to which some analysis problems remain tractable.

\paragraph{Parameterized complexity.}
A brief review will not be out of place here.
Let $\macalph$ be a finite alphabet in which instances $\macprino
\in \macalph^{*}$ of a problem $\macpro \subseteq \macalph^{*}$ are
specified, where $\macpro$ is the set of \macyes{} instances. The
complexity of a problem is stated in terms of the amount of
resources---space, time---needed by any algorithm solving it,
measured as a function of the size $|\macprino|$ of the problem instance. 
In parameterized complexity, introduced by Downey and Fellows \cite{DF99}, 
the dependence of resources needed is also measured in terms of a
parameter $\macparamo ( \macprino)$ of the input, which is usually
less than the input size $| \macprino|$. A parameterized problem is
said to be \name{fixed parameter tractable} (\macfpt) if it can be
solved by an algorithm with running time
$\macparamdepo(\macparamo(\macprino))\mathit{poly}(|\macprino|)$ where
$\macparamdepo$ is some computable function and $\mathit{poly}$ is a polynomial.
(Similarly, a \macparapsp{} algorithm \cite{FG03} is one that runs in space 
$\macparamdepo(\macparamo(\macprino))\mathit{poly}(|\macprino|)$.)

For example, consider the problem of checking that all strings
accepted by a given finite state automaton satisfy a given Monadic
Second Order (MSO) sentence. The size of an instance of this problem is
the sum of sizes of the automaton and the MSO sentence. If the size of
the MSO sentence is considered as a parameter, then this problem if
\macfpt{}, by B\"uchi, Elgot, Trakhtenbrot theorem \cite{Buchi}.

There is a  parameterized complexity class
\macwone{}, lowest in a hierarchy of intractable classes called the
\maccc{W}-hierarchy \cite{DF99} (similar to the polynomial time hierarchy). 
A parameterized problem
complete for $\macwone$ is to decide if there is an accepting
computation of at most $k$ steps in a given non-deterministic Turing
machine, where the parameter is $k$ \cite{DF99}. It is widely believed
that parameterized problems hard for \macwone{} are not \macfpt{}.
To prove that a problem is hard for a parameterized complexity
class, we have to give a \name{parameterized reduction} from a problem
already known to be hard to our problem. A parameterized reduction from
$(\macpro,\macparamo)$ to $(\macpro',\macparamo')$ is an algorithm
$\macredmap$ that maps problem instances in (resp.~outside) $\macpro$
to problem instances in (resp.~outside) $\macpro'$. There must be
computable functions $\macparamdepo$ and $\macparamdept$ and a polynomial
$p$ such that the algorithm $\macredmap$ on input $\macprino$ terminates 
in time $\macparamdepo(\macparamo(\macprino))p(|\macprino|)$ and
$\macparamo'(\macredmap(\macprino))\le
\macparamdept(\macparamo(\macprino))$, where $\macredmap ( \macprino)$
is the problem instance output by $\macredmap$.

\paragraph{Results.}
Demri, Laroussinie and Schnoebelen considered synchronized transition 
systems, a form of $1$-safe Petri nets \cite{DLS06} and showed that
the number of synchronizing components (processes) is not a parameter
which makes analysis tractable. Likewise, our first results are
negative. All parameters mentioned below are defined in
\secref{sec:preliminaries}.
\begin{itemize}
    %macugh
  \item With the pathwidth of the flow graph of the $1$-safe
    Petri net as parameter, reachability, coverability, Computational
    Tree Logic (CTL) and the complement of Linear Temporal Logic (LTL)
    model checking problems are all \macwone{}-hard, even when the
    size of the formula is a constant. In contrast, for the class of
    sequential transition systems and formula size as parameter,
    B\"uchi's theorem is that model checking for MSO logic is \macfpt.
  \item As a corollary, we also prove a conjecture of Downey, Fellows 
    and Stege that the \macsdp{} problem \cite[section 5]{DFS99} is 
    \macwone-hard when parameterized by treewidth.
  \item With the benefit depth of the $1$-safe Petri net as parameter,
    reachability, coverability, CTL and the complement of LTL model
    checking problems are \macwone{}-hard, even when the size of the
    formula is a constant.
\end{itemize}

We are luckier with our third parameter.
\begin{itemize}
    %macugh
  \item With the vertex cover number of the flow graph and 
    formula size as parameters, MSO model checking is \macfpt{}.
\end{itemize}

\paragraph{Perspective.}
As can be expected from the negative results, the class of 1-safe
Petri nets which are amenable to efficient analysis (i.e., those with
small vertex cover) is not too large.  But even for this class, a
reachability graph construction can be of exponential size, so just an
appeal to B\"uchi's theorem is not sufficient to yield our result.

Roughly speaking, our \macfpt\/ algorithm works well for systems which
have a small ``core'' (vertex cover), a small number of
``interface types'' with this core, but any number of component
processes using these interface types to interact with the core (see
\figref{fig:MSExample}).  Thus, we can have a large amount of conflict
and concurrency but a limited amount of causality.  Recall that
S-systems and T-systems have no concurrency and no conflict, respectively.
Since all we need from the logic is a procedure which produces an 
automaton from a formula, 
we are able to use the most powerful, MSO logic.
Our proofs combine the well known automata-theoretic method
\cite{Buchi,Vardi96,HAB97} with a powerful result about feasibility of
Integer Linear Programming (\macilp{}) parameterized by the number of
variables \cite{Lenstra83,Kannan87,FT87}.  

%A different way of viewing our results is that just as treewidth 
%is a parameter which provides an algorithmic metatheorem \cite{G07,K08}
%for properties of graphs, the stronger parameter vertex cover number 
%provides an algorithmic metatheorem for properties when we use 
%concurrency to succinctly represent systems.

\paragraph{Related work.} 
Drusinsky and Harel studied nondeterminism, alternation and concurrency 
in finite automata from a complexity point of view \cite{DH84}.
Their results also hold for 1-bounded Petri nets.

The \macsdp{} problem considered by Downey, Fellows and Stege
\cite{DFS99} can simulate Petri nets. They showed that with treewidth
and the length of the firing sequence as parameters, the reachability
problem is \macfpt{}. They conjectured that with treewidth alone as
parameter, the problem is \macwone{}-hard. 
%which we prove here. 

Fellows \emph{et al} showed that various graph layout problems that
are hard with treewidth as parameter (or whose complexity
parameterized by treewidth is not known) are \macfpt{} when
parameterized by vertex cover number \cite{FLMRS08}.  They also used
tractability of \macilp{} and extended feasibility to optimization.

\paragraph{Acknowledgements.} We thank the anonymous
\textsc{Concur} referees for providing detailed comments that helped in 
improving the presentation.

\section{Preliminaries}
\label{sec:preliminaries}

\subsection{Petri nets}
A Petri net is a 4-tuple $\macnet=(\macplaces,\mactranss,
\macpre,\macpost)$, $\macplaces$ a set of places, $\mactranss$ a set
of transitions, $\macpre:\macplaces\times \mactranss\to \{0,1\}$ (arcs
going from places to transitions) and $\macpost:\macplaces\times
\mactranss\to \{0,1\}$ (arcs going from transitions to places) the
incidence functions.  A place $\macplaceo$ is an \macpopterm{input}
(\macpopterm{output}) place of a transition $\mactranso$ if
$\macpre(\macplaceo,\mactranso) = 1$ ($\macpost(\macplaceo,\mactranso)
= 1$) respectively. We use $\macipls{ \mactranso}$ ($\macopls{
\mactranso}$) to denote the set of input (output) places of a
transition $\mactranso$. In diagrams, places are shown as circles and
transitions as thick bars. Arcs are shown as directed edges between
places and transitions.

Given a Petri net $\macnet$, we associate with it an undirected
\concept{flow graph} $\macgraph(\macnet)=(\macplaces,E)$ where
$(\macplaceo_{1},\macplaceo_{2}) \in E$ iff for some transition
$\mactranso$,
$\macpre(\macplaceo_{1},\mactranso)+\macpost(\macplaceo_{1},\mactranso)
\ge 1$ and
$\macpre(\macplaceo_{2},\mactranso)+\macpost(\macplaceo_{2},
\mactranso) \ge 1$. If a place $\macplaceo$ is both an input and an
output place of some transition, the vertex corresponding to
$\macplaceo$ has a self loop in $\macgraph(\macnet)$.

A \name{marking} $\macmark:\macplaces\to \macNat$ 
can be thought of as a configuration
of the Petri net, with each place $\macplaceo$ having
$\macmark(\macplaceo)$ tokens. 
We will only deal with \concept{1-safe} Petri nets in this paper,
where the range of markings is restricted to $\{0,1\}$. 
Given a Petri net $\macnet$ with a
marking $\macmark$ and a transition $\mactranso$ such that for every
place $\macplaceo$, $\macmark(\macplaceo)\ge
\macpre(\macplaceo,\mactranso)$, the transition $\mactranso$ is said
to be \name{enabled} at $\macmark$ and can be
\name{fired} (denoted $\macmark\macStep{\mactranso}\macmark'$) 
giving
$\macmark'(\macplaceo)=\macmark(\macplaceo)-\macpre(\macplaceo,\mactranso)
+ \macpost(\macplaceo,\mactranso)$ for every place $\macplaceo$.  
This is generalized to a \name{firing sequence}
$\macmark\macStep{\mactranso_{1}}\macmark_{1}\macStep{\mactranso_{2}}
\cdots\macStep{\mactranso_{\macfslidx}} \macmark_{\macfslidx}$, more briefly
$\macmark\macStep{\mactranso_{1} \mactranso_2 \cdots \mactranso_{\macfslidx}}
\macmark_{\macfslidx}$.
A firing sequence
$\macfirseqo$ enabled at $\macmark_{0}$ is said to be
\name{maximal} if it is infinite, or if 
$\macmark_{0} \macStep{\macfirseqo} \macmark$ and no transition is
enabled at $\macmark$.

\begin{definition}[Reachability, coverability]
  \label{def:ReachCover}
  Given a $1$-safe Petri net $\macnet$ with initial marking
  $\macmark_{0}$ and a target marking $\macmark: \macplaces \to
  \{0,1\}$, the reachability problem is to decide if there is a firing
  sequence $\macfirseqo$ such that $\macmark_{0} \macStep{\macfirseqo}
  \macmark$. The coverability problem is to decide if there is a
  firing sequence $\macfirseqo$ and some marking $\macmark':\macplaces
  \to \{0,1\}$ such that $\macmark_{0} \macStep{\macfirseqo}
  \macmark'$ and $\macmark'( \macplace) \ge \macmark( \macplace)$ for
  every place $\macplace$.
\end{definition}

\subsection{Logics}
Linear Temporal Logic (LTL) is a formalism in which many properties of
transition systems can be specified \cite[section 4.1]{ESPCompl98}.
We use the syntax of \cite{ESPCompl98}, in particular the places
$\macplaces$ are the atomic formulae.
%\begin{align*}
%  \macltlfo ::= \macplace \in
%  \macplaces~|~\lnot~\macltlfo~|~\macltlfo_{1}\land\macltlfo_{2}~|~
%  \macnext\macltlfo~|~\macltlfo_{1}\macuntil\macltlfo_{2}
%\end{align*}
The LTL formulas are interpreted on \name{runs}, sequences of markings 
$\macmarks = \macmark_{0} \macmark_{1} \cdots$ from a firing sequence
of a $1$-safe Petri net. 
The satisfaction of a LTL formula $\macltlfo$ at
some position $\macseqidx$ in a run is defined inductively,
in particular
%\begin{itemize}
%  \item 
$\macmarks,\macseqidx \models \macplace$ iff
    $\macmark_{\macseqidx}(\macplace) = 1$.
%  \item $\macmarks,\macseqidx \models \lnot \macltlfo$ iff
%    $\macmarks,\macseqidx \nvDash \macltlfo$.
%  \item $\macmarks,\macseqidx \models \macltlfo_{1} \land
%    \macltlfo_{2}$ iff $\macmarks,\macseqidx \models
%    \macltlfo_{1}$ and $\macmarks,\macseqidx \models
%    \macltlfo_{2}$.
%  \item $\macmarks, \macseqidx \models \macnext\macltlfo$ iff there
%    is a position $\macseqidx+1$ in $\macmarks$ and $\macmarks,
%    \macseqidx+1 \models \macltlfo$.
%  \item $\macmarks, \macseqidx \models \macltlfo_{1} \macuntil
%    \macltlfo_{2}$ iff for some $\macseqidxt \ge \macseqidx$,
%    $\macmarks, \macseqidxt \models \macltlfo_{2}$ and for all
%    $\macseqidx \le \macseqidxh < \macseqidxt$, $\macmarks,
%    \macseqidxh \models \macltlfo_{1}$.
%\end{itemize}
Much more expressive is the Monadic Second Order
(MSO) logic of B\"uchi \cite{Buchi}, interpreted on a maximal run
$\macmark_{0} \macmark_{1} \cdots$, with $\pi,s \models \macplace(x)$
iff $\macmark_{s(x)}(\macplace) = 1$ under an assignment $s$ to the
variables.  Boolean operations, first-order and monadic second-order 
quantifiers are available as usual.

Computational Tree Logic (CTL) is another logic that can be used to specify
properties of $1$-safe Petri nets. The reader is referred to
\cite[section 4.2]{ESPCompl98} for details.  

\begin{definition}[Model checking]
  \label{def:MC}
  Given a $1$-safe Petri net $\macnet$ with initial marking
  $\macmark_{0}$ and a logical formula $\macltlfo$, the model checking
  problem (for that logic) is to decide if for every maximal firing
  sequence $\macfirseqo$, the corresponding maximal run $\macmarks$
  satisfies $\macmarks, 0 \models \macltlfo$.  
\end{definition}

Reachability, coverability and LTL model checking for $1$-safe Petri nets 
are all \macpspace{}-complete \cite{ESPCompl98}.  
Habermehl gave an automata-theoretic model checking procedure for
Linear Time $\mu$-calculus on general Petri nets \cite{HAB97}.

\subsection{Parameters}
The study of parameterized complexity derived an initial motivation
from the study of graph parameters. Many \macnp{}-complete problems
can be solved in polynomial time on trees and are \macfpt{} on graphs
that have tree-structured decompositions.

\begin{definition}[Tree decomposition, treewidth, pathwidth]
A \name{tree decomposition} of a graph $\macgraph =
(\macvertexs, \macedges)$ is a pair
$(\mactree,(\macbag_{\mactn})_{\mactn\in \macnodes(\mactree)})$,
where $\mactree$ is a tree and $(\macbag_{\mactn})_{\mactn\in
\macnodes(\mactree)}$ is a family of subsets of $\macvertexs$ such
that:
\begin{itemize}
  \item For all $\macverto\in \macvertexs$, the set $\{\mactn\in
    \macnodes(\mactree)\mid \macverto\in \macbag_{\mactn}\}$ is
    nonempty and connected in $\mactree$.
  \item For every edge $(\macverto_{1},\macverto_{2})\in \macedges$,
    there is a $\mactn\in \macnodes(\mactree)$ such that
    $\macverto_{1},\macverto_{2} \in \macbag_{\mactn}$.
\end{itemize}
The width of such a decomposition is the number
$\max\{|\macbag_{\mactn}|\mid \mactn\in \macnodes(\mactree)\}-1$. The
\concept{treewidth} $\mactw(\macgraph)$ of $\macgraph$ is the minimum of
the widths of all tree decompositions of $\macgraph$. 
If the tree $\mactree$ in the definition of tree decomposition is a
path, we get a path decomposition. The \concept{pathwidth}
$\text{pw}(\macgraph)$ of $\macgraph$ is the minimum of the widths of
all path decompositions of $\macgraph$. 
\end{definition}

From the definition, it is clear that pathwidth is at least as large as
treewidth and any problem that is \macwone{}-hard with pathwidth as
parameter is also \macwone{}-hard with treewidth as parameter. 
%We prove lower bound with pathwidth as parameter, which is stronger than
%taking treewidth as parameter.
%An intuitive way to
%understand the role of treewidth on complexity of problems is through
%its relation to grids. A grid is a graph of the form shown below.
%\begin{center}
%  \includegraphics{PWGrid}
%\end{center}
%A large grid has large treewidth. A graph of small treewidth can not
%have a large grid embedded within it (formally, the graph can not have
%large grid minors; see \cite[Section 7.1]{DF99} for details). Recall
%that \macnp{}-complete problems can be reduced to coloring problems on
%grids. 
A fundamental result by Courcelle \cite{C90} shows that graphs of
small treewidth are easier to handle algorithmically: checking whether
a graph satisfies a MSO sentence is \macfpt{} if the graph's treewidth
and the MSO sentence's length are parameters. In our context, the
state space of a concurrent system can be considered a graph. However,
due to the state explosion problem, the state space can be very large.
Instead, we impose treewidth restriction on a compact representation
of the large state space --- a $1$-safe Petri net. Note also 
that we are not model checking the state space itself but only the
language of words generated by the Petri net. 

\begin{definition}[Vertex cover number]
A \name{vertex cover} $\macvertcov\subseteq \macvertexs$ of a graph
$\macgraph=(\macvertexs,\macedges)$ is a subset of vertices such that
for every edge in $E$, at least one of its vertices is in
$\macvertcov$. The \concept{vertex cover number} of $\macgraph$ is 
the size of a smallest vertex cover. 
\end{definition}

\begin{definition}[Benefit depth \cite{PL09}]
The set of
places $\macben(\macplace)$ \name{benefited} by a place $\macplace$ is
the smallest set of places (including $\macplace$) such that any
output place of any output transition of a place in
$\macben(\macplace)$ is also in $\macben(\macplace)$.  The
\concept{benefit depth} of a Petri net is defined as $\max_{\macplace
\in \macplaces}\{|\macben(\macplace)|\}$.  
\end{definition}

Benefit depth can be
thought of as a generalization of the out-degree in directed graphs.
For a Petri net, we take vertex covers of its flow
graph $\macgraph(\macnet)$. Any vertex cover of $\macgraph(\macnet)$
should include all vertices that have self loops.
It was shown in \cite{PL09,PraveenVC10} that benefit depth and vertex cover
number bring down the
complexity of coverability and boundedness in general Petri nets from
exponential space-complete \cite{RCK78} to \macparapsp{}.

\section{Lower bounds for $1$-safe Petri nets and pebbling}
\label{sec:twhardness}
\subsection{$1$-safe Petri nets, treewidth and pathwidth}
Here we prove \macwone{}-hardness of reachability in $1$-safe Petri
nets with the pathwidth of the flow graph as parameter, through a
parameterized reduction from the parameterized Partitioned Weighted
Satisfiability (\macppwcnfsat{}) problem.  The \name{primal graph} of
a propositional CNF formula has one vertex for each propositional
variable, and an edge between two variables iff they occur together in
a clause.  An instance of \macppwcnfsat{} problem is a triple
$(\macpcnff,\macpart:\macpv\to \{1,\dots,\macnumpart\},
\mactget:\{1,\dots,\macnumpart\}\to \macNat)$, where $\macpcnff$ is a
propositional CNF formula, $\macpart$ partitions the set of
propositional variables $\macpv$ into $\macnumpart$ parts and we need
to check if there is a satisfying assignment that sets exactly
$\mactget(\macpartidx)$ variables to $\top$ in each part
$\macpartidx$. Parameters are $\macnumpart$ and the pathwidth of the
primal graph of $\macpcnff$. We showed in an earlier paper that
\macppwcnfsat\/ is \macwone-hard when parameterized by the number of
parts $\macnumpart$ and the pathwidth of the primal graph \cite[Lemma
6.1]{PraveenMT10}.

Now we will demonstrate a parameterized reduction from \macppwcnfsat{}
to reachability in $1$-safe Petri nets, with the pathwidth of the flow
graph as parameter. Given an instance of \macppwcnfsat{}, let
$\macpvo_{1},\dots,\macpvo_{\macnumpv}$ be the variables used.
Construct an optimal path decomposition of the primal graph of the CNF
formula in the given \macppwcnfsat{} instance (doing this is
\macfpt{} \cite{BK96}). For every clause in the CNF formula, the
primal graph contains a clique formed by all variables occurring in
that clause. There will be at least one bag in the path decomposition
of the primal graph that contains all vertices in this clique
\cite[Lemma 6.49]{DF99}.  Order the bags of the path decomposition
from left to right and call the clause whose clique appears first
$\macclause_{1}$, the clause whose clique appears second as
$\macclause_{2}$ and so on.  If more than one such such clique appear
for the first time in the same bag, order the corresponding clauses
arbitrarily. Let $\macclause_{1},\dots,\macclause_{\macnumcl}$ be the
clauses ordered in this way. We will call this the \name{path
decomposition ordering} of clauses, and use it to prove that the
pathwidth of the flow graph of the constructed $1$-safe Petri net is
low (\lemref{lem:TgetCntNetPWLow}). For a partition $\macpartidx$
between $1$ and $\macnumpart$, we let $\macnumpv[\macpartidx]$ be the
number of variables in $\macpartidx$.  Following are the places of our
$1$-safe Petri net.
\begin{enumerate}
  \item For every propositional variable $\macpvo_{\macpvidx}$ used in
    the given \macppwcnfsat{} instance, places $\macpvo_{\macpvidx},
    \mactpl_{\macpvidx}, \macfpl_{\macpvidx}$.
  \item For every partition $\macpartidx$ between $1$ and $\macnumpart$,
    places $\macpartindo^{\macpartidx}, \macpartindt^{\macpartidx},
    \macpartctro_{\macpartidx}^{0},\dots,\macpartctro_{\macpartidx}^{\mactget(\macpartidx)}$
    and
    $\macpartctrt_{\macpartidx}^{0},\dots,\macpartctrt_{\macpartidx}^{\macnumpv[\macpartidx]-
    \mactget(\macpartidx)}$.
  \item For each clause $\macclause_{\macclidx}$, a place
    $\macclpl_{\macclidx}$. 
  Additional places $\macclause_{\macnumcl+1}$, $\macserpl$,
    $\macglpl$.
\end{enumerate}

The construction of the Petri net is illustrated in the following
diagrams. The notation $\macpart(\macpvidx)$ stands for the partition
to which $\macpvo_{\macpvidx}$ belongs. 
Intuitively, the truth assignment of $\macpvo_{\macpvidx}$ is
determined by firing $\macttr_{\macpvidx}$ or $\macftr_{\macpvidx}$ in
\figref{fig:PWVarAssigNet}.  The token in
$\mactpl_{\macpvidx}/\macfpl_{\macpvidx}$ is used to check
satisfaction of clauses later. The token in
$\macpartindo^{\macpart(\macpvidx)}/\macpartindt^{\macpart(\macpvidx)}$ is used to count the
number of variables set to $\top/\bot$ in each part, with the part of
the net in \figref{fig:PWTgCnt}. 
For each clause $\macclause_{\macclidx}$ between $1$ and $\macnumcl$,
the part of the net shown in \figref{fig:PWClSat} is constructed. In
\figref{fig:PWClSat}, it is assumed that $\macclause_{\macclidx} =
\macpvo_{1} \lor \overline{\macpvo_{2}} \lor \macpvo_{3}$. 
Intuitively, a token can be moved from place $\macclpl_{\macclidx}$ to
$\macclpl_{\macclidx+1}$ iff the clause $\macclause_{\macclidx}$ is
satisfied by the truth assignment determined by the firings of
$\macttr_{\macpvidx}/\macftr_{\macpvidx}$ for each $\macpvidx$
\begin{figure}[ht]
  \begin{center}
    \includegraphics{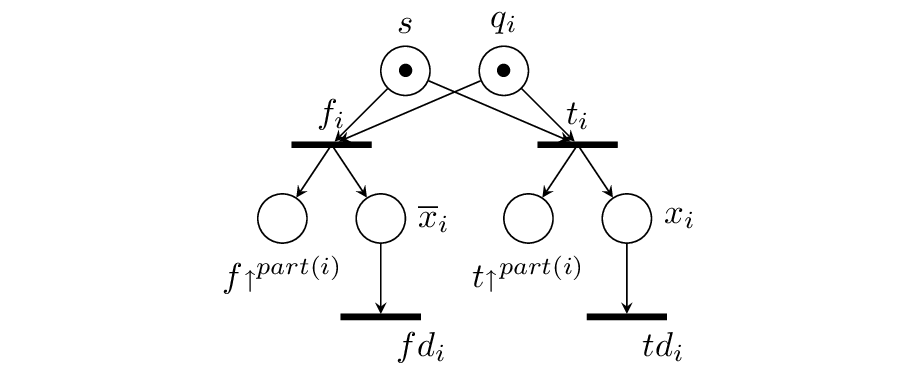}
  \end{center}
  \caption{Part of the net for each variable $\macpvo_{\macpvidx}$}
  \label{fig:PWVarAssigNet}
\end{figure}
\begin{figure}[!h]
  \begin{center}
    \includegraphics{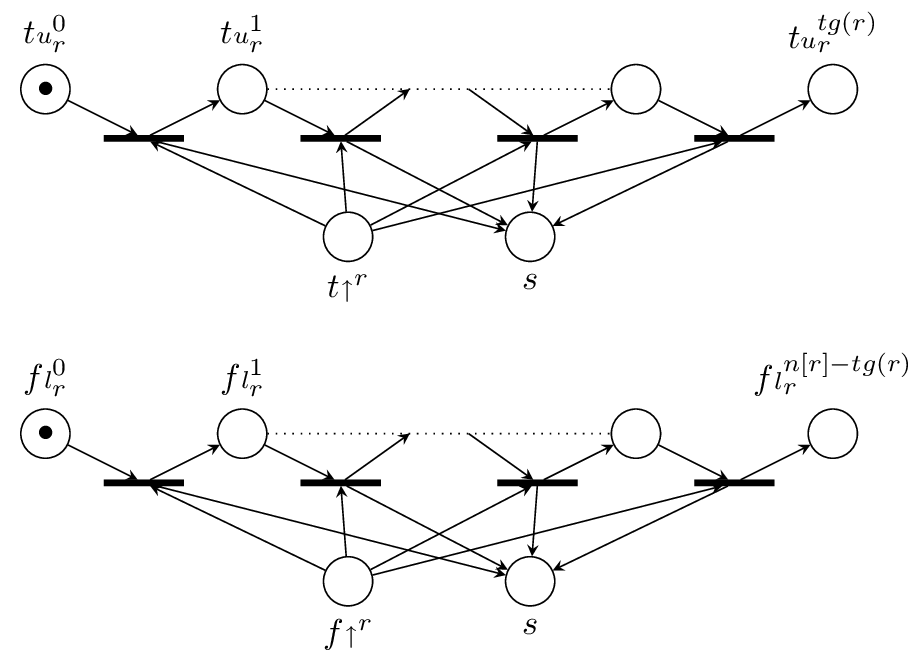}
  \end{center}
  \caption{Part of the net for each part $\macpartidx$ between $1$ and
  $\macnumpart$}
  \label{fig:PWTgCnt}
\end{figure}
between
$1$ and $\macnumpv$. The net in \figref{fig:PWGlCheck} checks that the
target has been met in all partitions.

The initial marking of the constructed net consists of $1$ token each
in the places $\macpvo_{1},\dots,\macpvo_{\macnumpv}, \macserpl,
\macpartctro_{1}^{0}, \dots, \macpartctro_{\macnumpart}^{0},
\macpartctrt_{1}^{0}, \dots, \macpartctrt_{\macnumpart}^{0}$ and
$\macclpl_{1}$, with $0$ tokens in all other places. The final marking
to be reached has a token in the places $\macserpl$ and $\macglpl$.
%(and none elsewhere).

\begin{lemma}
  \label{lem:CNFSatPNReach}
  Given a \macppwcnfsat{} instance, constructing the Petri net
  described above  is \macfpt{}. The constructed Petri net is
  $1$-safe. The given instance of \macppwcnfsat{} is a \macyes{}
  instance iff in the constructed $1$-safe net, the required final
  marking can be reached from the given initial marking.
\end{lemma}
\begin{fver}
\begin{proof}
  The only non-trivial process in the construction of the Petri net is
  computing an optimal path decomposition of the primal graph of the
  CNF formula in the given \macppwcnfsat{} instance. Doing this is
  \macfpt{} and the rest of the construction can be done in polynomial
  time. It is not difficult to see that the constructed net is
  $1$-safe.

\begin{figure}[ht]
  \begin{center}
    \includegraphics{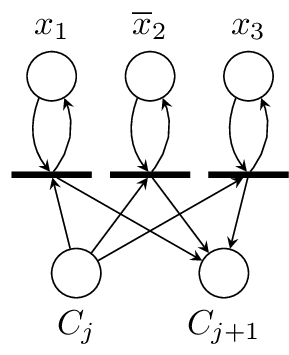}
  \end{center}
  \caption{Part of the net for each clause $\macclause_{\macclidx}$}
  \label{fig:PWClSat}
\end{figure}
\begin{figure}[h]
  \begin{center}
    \includegraphics{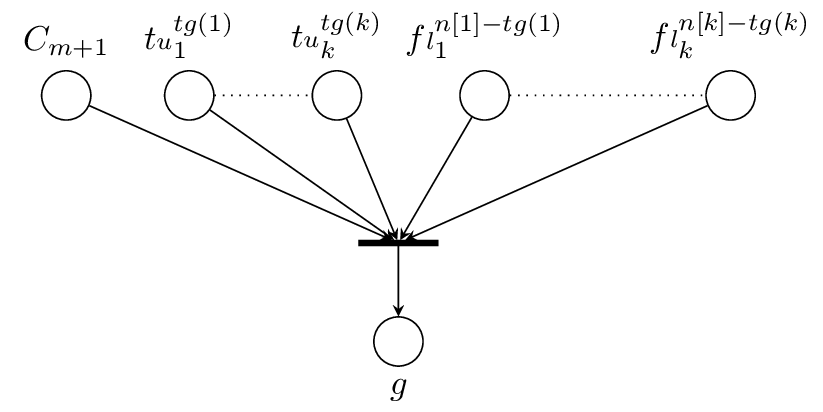}
  \end{center}
  \caption{Part of the net to check that target has been met}
  \label{fig:PWGlCheck}
\end{figure}

  Suppose the given instance of \macppwcnfsat{} is a \macyes{}
  instance. Starting with $\macpvidx=1$, repeat the following firing
  sequence for each $\macpvidx$ between $1$ and $\macnumpv$. If
  $\macpvo_{\macpvidx}$ is $\top$ in the witnessing satisfying
  assignment, fire $\macttr_{\macpvidx}$ else fire
  $\macftr_{\macpvidx}$. Then use the token thus put into
  $\macpartindo^{\macpart(\macpvidx)}/
  \macpartindt^{\macpart(\macpvidx)}$ respectively to shift a token
  one place to the right in \figref{fig:PWTgCnt} and put a token back
  in the place $\macserpl$. Continue with the next $\macpvidx$. Since
  the witnessing assignment meets the target in each
  partition, we will have one token each in the places
  $\macpartctro_{1}^{\mactget(1)}, \dots,
  \macpartctro_{\macnumpart}^{\mactget(\macnumpart)},
  \macpartctrt_{1}^{\macnumpv[1]-\mactget(1)}, \dots,
  \macpartctrt_{\macnumpart}^{\macnumpv[\macnumpart]-\mactget(\macnumpart)}$.
  In addition, there will be a token in $\mactpl_{\macpvidx} /
  \macfpl_{\macpvidx}$ iff the witnessing assignment set
  $\macpvo_{\macpvidx}$ to $\top / \bot$ respectively. Since this
  witnessing assignment satisfies all the clauses of the CNF formula,
  we can move the initial token in $\macclpl_{1}$ to
  $\macclpl_{\macnumcl+1}$ using the transitions in
  \figref{fig:PWClSat}. Now, the transition in \figref{fig:PWGlCheck}
  can be fired to get a token into the place $\macglpl$. Now, the only
  tokens left are those in the places $\macserpl$ and $\macglpl$, and
  those in $\mactpl_{\macpvidx} / \macfpl_{\macpvidx}$. We can remove
  the tokens in $\mactpl_{\macpvidx} / \macfpl_{\macpvidx}$ by firing
  $\mactdel_{\macpvidx} / \macfdel_{\macpvidx}$ to reach the 
  final marking.

  Suppose the required final marking is reachable in the constructed
  Petri net. Since a token has to be added to the place $\macglpl$ to
  reach the final marking and the transition in \figref{fig:PWGlCheck}
  is the only transition that can add tokens to $\macglpl$, all input
  places of that transition must receive a token. The only way to get
  a token in places
  $\macpartctro_{\macpartidx}^{\mactget(\macpartidx)}$ is to shift the
  initial token in the place $\macpartctro_{\macpartidx}^{0}$
  $\mactget(\macpartidx)$ times. This requires exactly
  $\mactget(\macpartidx)$ tokens in the place
  $\macpartindo^{\macpartidx}$. A similar argument holds for getting a
  token in
  $\macpartctrt_{\macpartidx}^{\macnumpv[\macpartidx]-\mactget
  (\macpartidx)}$. Since the only way to add a token to
  $\macpartindo^{\macpartidx} / \macpartindt^{\macpartidx}$ is to fire
  transitions $\macttr_{\macpvidx} / \macftr_{\macpvidx}$ (such that
  $\macpart(\macpvidx) = \macpartidx$), the only
  way to get a token each in $\macpartctro_{1}^{\mactget(1)}, \dots,
  \macpartctro_{\macnumpart}^{\mactget(\macnumpart)},
  \macpartctrt_{1}^{\macnumpv[1]-\mactget(1)}, \dots,
  \macpartctrt_{\macnumpart}^{\macnumpv[\macnumpart]-\mactget(\macnumpart)}$
  is to fire either $\macttr_{\macpvidx}$ or $\macftr_{\macpvidx}$
  for each $\macpvidx$ between $1$ and $\macnumpv$. Consider
  any firing sequence reaching the required final marking. Consider
  the truth assignment to $\macpvo_{1}, \dots, \macpvo_{\macnumpv}$
  that assigns $\top$ to exactly those variables $\macpvo_{\macpvidx}$
  such that $\macttr_{\macpvidx}$ was fired in the firing sequence.
  This truth assignment meets the target for each part since this firing
  sequence adds one token each to the places
  $\macpartctro_{1}^{\mactget(1)}, \dots,
  \macpartctro_{\macnumpart}^{\mactget(\macnumpart)},
  \macpartctrt_{1}^{\macnumpv[1]-\mactget(1)}, \dots,
  \macpartctrt_{\macnumpart}^{\macnumpv[\macnumpart]-\mactget(\macnumpart)}$.
  To reach the final marking, a token is also required at the place
  $\macclpl_{\macnumcl+1}$. The only way to get this token is to shift
  the initial token in $\macclpl_{1}$ to $\macclpl_{\macnumcl+1}$
  through the transitions in \figref{fig:PWClSat}. This means that
  every clause is satisfied by the truth assignment we constructed.
\qed
\end{proof}
\end{fver}

It remains to prove that the pathwidth of the flow graph of
the constructed $1$-safe net is a function of the parameters of the
\macppwcnfsat{} instance.
\begin{lemma}
  \label{lem:TgetCntNetPWLow}
  Suppose a given instance of \macppwcnfsat{} has a CNF formula whose
  primal graph has pathwidth $\macpw$ and $\macnumpart$ parts.  Then,
  the flow graph of the $1$-safe net constructed as described above
  has pathwidth at most $3\macpw+ 4\macnumpart+ 7$.
\end{lemma}
\begin{fver}
\begin{proof}
  We show a path decomposition of the flow
  graph of the net. Call the set of places
  $\{\macserpl, \macglpl, \macclpl_{\macnumcl+1},
  \macpartindo^{1}, \dots, \macpartindo^{\macnumpart},
  \macpartindt^{1}, \dots, \macpartindt^{\macnumpart},
  \macpartctro_{1}^{\mactget(1)},\dots,
  \macpartctro_{\macnumpart}^{\mactget(\macnumpart)},
  \macpartctrt_{1}^{\macnumpv[1]-\mactget(1)}, \dots,\\
  \macpartctrt_{\macnumpart}^{\macnumpv[\macnumpart]-
  \mactget(\macnumpart)}\}$ as $\macpartinds$.  
  Consider an optimal path decomposition of the primal graph of the
  CNF formula. In every bag, replace every occurrence of each
  $\macpvo_{\macpvidx}$ by the set $\{\macpvo_{\macpvidx},
  \mactpl_{\macpvidx}, \macfpl_{\macpvidx}\} \cup \macpartinds$.

  Let $\macclause_{1},\dots, \macclause_{\macnumcl}$ be the clauses in
  the path decomposition order as explained in the beginning of this
  sub-section. We will first show that places representing clauses can
  be added to the bags of the above decomposition without increasing
  their size much, while maintaining the invariant that all bags
  containing any one place are connected in the decomposition. We will
  do this by \name{augmenting} existing bags with new elements: if
  $\macbag$ is any bag in the decomposition and $\macplace$ is an
  element not in $\macbag$, augmenting $\macbag$ with $\macplace$
  means creating a new bag $\macbag'$ immediately to the left of
  $\macbag$ containing $\macplace$ in addition to the elements in
  $\macbag$. We will call the new bag $\macbag'$ thus created an
  augmented bag. Perform the following operation in increasing order
  for each $\macclidx$ between $1$ and $\macnumcl$: if $\macbag$ is
  the first non-augmented bag from left to contain all literals of the
  clause $\macclause_{ \macclidx}$, augment $\macbag$ with $\macclpl_{
  \macclidx}$.

  There will be $\macnumcl$ new bags created due to the above
  augmentation steps. Due to the path decomposition ordering of
  $\macclause_{1}, \dots, \macclause_{ \macnumcl}$, the augmented bag
  containing $\macclpl_{ \macclidx +1}$ occurs to the right of the
  augmented bag containing $\macclpl_{ \macclidx}$ for each
  $\macclidx$, $1 \le \macclidx < \macnumcl$. There might be some
  non-augmented bags between the augmented bags containing
  $\macclpl_{\macclidx}$ and $\macclpl_{ \macclidx + 1}$. If so, add
  $\macclpl_{ \macclidx}$ to such non-augmented bags.  Now, to every
  bag, if it contains $\macclpl_{\macclidx}$ for some $\macclidx$
  between $1$ and $\macnumcl$, add $\macclpl_{\macclidx+1}$. It is
  routine to verify the following properties of the sequence we have
  with the bags modified as above.
  \begin{itemize}
    \item Each bag has at most $3\macpw+ 4\macnumpart+ 8$ elements.
    \item The set of bags containing any one element forms a
      contiguous sub-sequence.
    \item Every vertex and edge in any subgraph induced by the parts
      of the net in \figref{fig:PWVarAssigNet}, \figref{fig:PWClSat},
      and \figref{fig:PWGlCheck} is contained in some bag.
  \end{itemize}
  To account for the subgraph induced by the parts of the net in
  \figref{fig:PWTgCnt}, we append the following sequence of bags for each
  $\macpartidx$ between $1$ and $\macnumpart$:
  \begin{align*}
    &(\{\macpartctro_{\macpartidx}^{0},
    \macpartctro_{\macpartidx}^{1}\} \cup \macpartinds) -
    (\{\macpartctro_{\macpartidx}^{1},
    \macpartctro_{\macpartidx}^{2}\} \cup \macpartinds) - \cdots -
    (\{\macpartctro_{\macpartidx}^{\mactget(\macpartidx)-1},
    \macpartctro_{\macpartidx}^{\mactget(\macpartidx)}\} \cup
    \macpartinds) - \\
    &(\{\macpartctrt_{\macpartidx}^{0},
    \macpartctrt_{\macpartidx}^{1}\} \cup \macpartinds) -
    (\{\macpartctrt_{\macpartidx}^{1},
    \macpartctrt_{\macpartidx}^{2}\} \cup \macpartinds) - \cdots -
    (\{\macpartctrt_{\macpartidx}^{\macnumpv[\macpartidx]-\mactget(\macpartidx)-1},
    \macpartctrt_{\macpartidx}^{\macnumpv[\macpartidx]-\mactget(\macpartidx)}\} \cup
    \macpartinds)
  \end{align*}
  The resulting sequence of bags is a path decomposition of the
  flow graph of the Petri net, whose width is at most $3\macpw+
  4\macnumpart+ 7$.
\qed
\end{proof}
\end{fver}
In the above reduction, it is enough to check if in the constructed
$1$-safe net, we can reach a marking that has a token at the place
$\macglpl$. This can be expressed as reachability, coverability etc.
Hence we get:

\begin{theorem}
  \label{thm:PNReachPWHard}
  With the pathwidth (and hence treewidth also) of the flow
  graph of a $1$-safe Petri net as parameter, reachability,
  coverability, CTL model checking and the complement of LTL/MSO model
  checking (with formulas of constant size) are \macwone{}-hard.
\end{theorem}

\subsection{Graph pebbling problems, treewidth and pathwidth}
The techniques used in the above lower bound proof can be easily
translated to some graph pebbling problems \cite{DFS99}. 
As conjectured in \cite[section 5]{DFS99}, we prove that
\macsdp{} I, parameterized by treewidth is \macwone{}-hard. 
%In \cite{DFS99}, this problem is referred to as \macsdp{} I.
An instance of this problem has a bipartite digraph
$\macdgraph=(\macvertexs, \macarcs)$ for which the vertex set
$\macvertexs$ is partitioned $\macvertexs = \macred \cup \macblue$,
and also the arc set $\macarcs$ is
partitioned into two partitions $\macarcs = \macarcs^{+} \cup
\macarcs^{-}$.  The problem is to reach the \emph{finish state} where
there are pebbles on all the red vertices, starting from a \emph{start
state} where there are no pebbles on any of the red vertices, by a
series of moves of the following form:
\begin{itemize}
\item
  If $\macbvertex$ is a blue vertex such that for
  all $\macvertext$ such that $(\macvertext, \macbvertex) \in
  \macarcs^{+}$, $\macvertext$ is pebbled, and for all $\macvertext$
  such that $(\macvertext, \macbvertex) \in \macarcs^{-}$,
  $\macvertext$ is not pebbled (in which case we say that
  $\macbvertex$ is \emph{enabled}), then the set of vertices
  $\macvertext$ such that $(\macbvertex, \macvertext) \in
  \macarcs^{+}$ are reset by making them all pebbled, and the set of
  all vertices $\macvertext$ such that $(\macbvertex, \macvertext) \in
  \macarcs^{-}$ are reset by making them all unpebbled.
\end{itemize}

\begin{corollary}
  \label{cor:GraphPebblingHard}
  Parameterized by pathwidth (and hence by treewidth also), \macsdp{}
  is \macwone{}-hard.
\end{corollary}
\begin{cver}
  The proof is by a parameterized reduction from \macppwcnfsat{} to
  reachability in $1$-safe Petri nets as in the last sub-section, and
  another reduction from reachability in $1$-safe Petri nets to
  \macsdp{}.
\end{cver}
\begin{fver}
  \begin{proof}
    To reduce \macppwcnfsat{} to \macsdp{}, we first reduce the given
  \macppwcnfsat{} instance to a $1$-safe net as shown in 
  \lemref{lem:CNFSatPNReach}.
  From this $1$-safe net, construct an instance of \macsdp{} as follows.
  Let the set of all places form the set of vertices $\macred$ and
  the set of all transitions form the set of vertices $\macblue$.
  The arcs of the \macsdp{} instance are as follows.
  \begin{enumerate}
    \item If $\macpre(\macplace,\mactrans) = 1$ in the $1$-safe net,
      draw an $\macarcs^{+}$ arc from $\macplace$ to $\mactrans$ in
      the \macsdp{} instance.
    \item If $\macpre(\macplace,\mactrans) =1$ and
      $\macpost(\macplace, \mactrans) =0$, draw an
      $\macarcs^{-}$ arc from $\mactrans$ to $\macplace$.
    \item If $\macpre(\macplace,\mactrans) = 0$ and
      $\macpost(\macplace, \mactrans) =1$, draw an
      $\macarcs^{+}$ arc from $\mactrans$ to $\macplace$.
  \end{enumerate}
  Suppose that in the $1$-safe net,
  $\macmark_{1} \macStep{\mactrans} \macmark_{2}$. It is clear that
  the constructed \macsdp{} instance in the state where precisely
  those red vertices are pebbled that have a token in $\macmark_{1}$
  enables the blue vertex $\mactrans$, and can move to the state where
  precisely those red vertices are pebbled that have a token in
  $\macmark_{2}$. Add a special blue vertex $\macbvertex_{1}$ with
  $\macarcs^{+}$ arcs from $\macbvertex_{1}$ to $\macpvo_{1},
  \macpvo_{2}, \dots, \macpvo_{\macnumpv}$, $\macpartctro_{1}^{0},
  \dots, \macpartctro_{\macnumpart}^{0}$, $\macpartctrt_{1}^{0}, \dots,
  \macpartctrt_{\macnumpart}^{0}$, $\macclpl_{1}$ and $\macserpl$. Add
  $\macarcs^{-}$ arcs from all red vertices to $\macbvertex_{1}$. In
  the start state where there no pebbles at all, $\macbvertex_{1}$ is
  the only blue vertex enabled. The blue vertex $\macbvertex_{1}$ is
  enabled only in the start state. Upon performing the legal move
  using $\macbvertex_{1}$ from the start state, we will reach a state
  in which precisely those red vertices are pebbled that have a token
  in the initial marking of the $1$-safe net. From this state, there
  is at least one pebbled red vertex in any reachable state, so
  $\macbvertex_{1}$ is never enabled again. From this state, we can
  reach a state with the red vertex $\macglpl$ pebbled iff the given
  \macppwcnfsat{} instance is a \macyes{} instance. Add another
  special blue vertex $\macbvertex_{2}$ with an $\macarcs^{+}$ arc
  from the red vertex $\macglpl$ to $\macbvertex_{2}$. Add
  $\macarcs^{+}$ arcs from $\macbvertex_{2}$ to all red vertices. All
  blue vertices except $\macbvertex_{1}$ and $\macbvertex_{2}$
  unpebble at least one red vertex. Hence, the only way to reach the
  finish state (where all red vertices must be pebbled) from the start
  state is to enable $\macbvertex_{2}$. The only way to enable
  $\macbvertex_{2}$ is to reach a state where the red vertex
  $\macglpl$ is pebbled. Hence, the constructed \macsdp{} instance is
  a \macyes{} instance iff the given \macppwcnfsat{} instance is a
  \macyes{} instance.

  To complete the reduction, it only remains to show that the
  pathwidth of the \macsdp{} instance is bounded by the pathwidth of
  %macugh
  the flow graph of the intermediate $1$-safe net. Consider an
  %macugh
  optimal path decomposition of this flow graph. For every
  transition $\mactrans$, the set of all input and output places of
  %macugh
  $\mactrans$ forms a clique in the flow graph. Hence, there
  will be at least one bag $\macbag$ in the path decomposition
  containing all these places.  Create an extra bag $\macbag'$
  adjacent to $\macbag$ containing all elements of $\macbag$ and also
  the blue vertex corresponding to $\mactrans$. After doing this for
  each transition, add the vertices $\macbvertex_{1}$ and
  $\macbvertex_{2}$ to all bags. The resulting decomposition is a path
  decomposition of the \macsdp{} instance. Its width is at most $3$
  %macugh
  more than the pathwidth of the flow graph of the $1$-safe net.
\qed
  \end{proof}
\end{fver}

\subsection{$1$-safe Petri nets and benefit depth}
Here we show that the parameter benefit depth is not helpful for
$1$-safe Petri nets, by showing \macwone{}-hardness using a
parameterized reduction from the constraint satisfaction problem
(\maccsp).

\begin{theorem}
  \label{thm:BD1SafeReachHard}
  With benefit depth as the parameter in $1$-safe Petri nets,
  reachability, coverability, CTL model checking and the complement of
  the LTL/MSO model checking problems, even with formulas of constant
  size, are \macwone{}-hard.
\end{theorem}
\begin{fver}
  The rest of this section is devoted to a proof of the above theorem.
To show that with benefit depth as parameter, reachability in $1$-safe
nets is \macwone{}-hard, we will show a \macfpt{} reduction from the
constraint satisfaction problem (\maccsp). With the size of the domain
$\macdom$ and the maximum number of constraints in which any one
variable can occur (called degree) $\macdeg$ as parameters, \maccsp{}
is \macwone{}-hard \cite[Corollary 2]{SS10}. Given an instance of
\maccsp{} with domain size $\macdom$, degree $\macdeg$, $\macnumvar$
variables and $\macnumconstr$ constraints, we construct a $1$-safe net
with the following places.
\begin{enumerate}
  \item For every variable $\macvaro_{\macvaridx}$, a place
    $\macvaro_{\macvaridx}$.
  \item For every constraint $\macconstro_{\macconstridx}$ where
    $\macconstridx$ is between $1$ and $\macnumconstr$, a place
    $\macconstro_{\macconstridx}$.
  \item For every $\macvaridx$ between $1$ and $\macnumvar$, for every
    domain element $\macdomidx$ between $1$ and $\macdom$, for every
    constraint $\macconstro_{ \macconstridx}$ in which
    $\macvaro_{ \macvaridx}$ appears, the place
    $\macvaro[\macvaridx]_{\macconstridx}^{\macdomidx}$.
  \item One place $\macglpl$ for checking that all constraints are
    satisfied.
\end{enumerate}
We assume without loss of generality that every variable occurs in at
least one constraint. Construction of the $1$-safe net is illustrated
in the following diagrams. For every variable $\macvaro_{\macvaridx}$
and domain value $\macdomidx$ (between $1$ and $\macdom$), part of the
net shown in \figref{fig:BDVarSet} is constructed.
\begin{figure}[!tp]
  \begin{center}
    \includegraphics{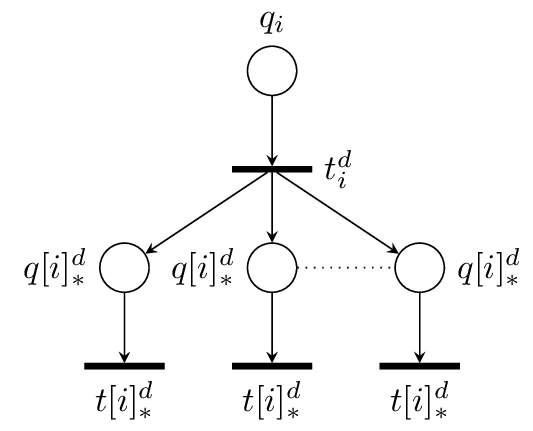}
  \end{center}
  \caption{Part of the net for every variable $\macvaro_{\macvaridx}$
  and domain value $\macdomidx$}
  \label{fig:BDVarSet}
\end{figure}
Intuitively, the transition $\mactrans_{\macvaridx}^{\macdomidx}$ is
fired to assign domain value $\macdomidx$ to $\macvaro_{\macvaridx}$.
In \figref{fig:BDVarSet}, the set of places labelled by
$\macvaro\macsbr{\macvaridx}_{*}^{\macdomidx}$ should be understood to
stand for the set of places
$\{\macvaro\macsbr{\macvaridx}_{\macconstridx}^{\macdomidx} \mid
\macvaro_{\macvaridx} \text{ occurs in constraint }
\macconstro_{ \macconstridx}\}$. The set of transitions labelled
$\mactrans\macsbr{\macvaridx}_{*}^{\macdomidx}$ should be similarly
understood.

For every constraint $\macconstro_{\macconstridx}$ and every
admissible tuple of domain values for $\macconstro_{\macconstridx}$,
part of the net shown in \figref{fig:BDConCheck} is constructed.
\begin{figure}[!tp]
  \begin{center}
    \includegraphics{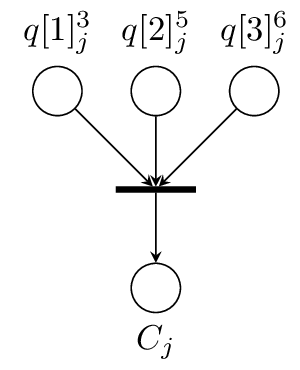}
  \end{center}
  \caption{Part of the net for every constraint
  $\macconstro_{\macconstridx}$ and every admissible tuple}
  \label{fig:BDConCheck}
\end{figure}
In \figref{fig:BDConCheck}, it is assumed that the constraint
$\macconstro_{\macconstridx}$ consists of variables
$\macvaro_{1}, \macvaro_{2}$ and $\macvaro_{3}$ and that $(3,5,6)$ is
an admissible tuple for this constraint. Finally, the part of the net in
\figref{fig:BDGlCheck} verifies that all constraints are satisfied.
\begin{figure}[!ht]
  \begin{center}
    \includegraphics{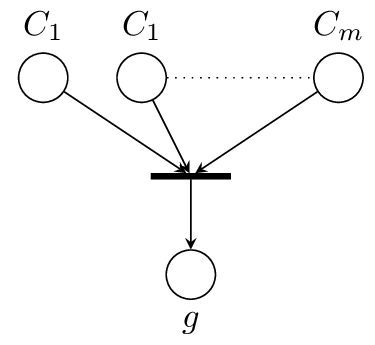}
  \end{center}
  \caption{Part of the net to check that all constraints are satisfied}
  \label{fig:BDGlCheck}
\end{figure}
The initial marking has $1$ token each in each of the places
$\macvaro_{1}, \dots, \macvaro_{\macnumvar}$ and $0$ tokens in all
other places. The final marking to be reached is $1$ token at the
place $\macglpl$ and $0$ tokens in all other places.

\begin{lemma}
  \label{lem:CSP1SafeNet}
  Given a \maccsp{} instance of domain size $\macdom$ and degree
  $\macdeg$, the benefit depth of the $1$-safe net constructed above
  is at most $2+ \macdeg(\macdom+1)$. The given \maccsp{} instance is
  satisfiable iff the required final marking is reachable from the
  initial marking in the constructed $1$-safe net.
\end{lemma}
\begin{proof}
  Maximum number of places are benefited by some place in
  $\{\macvaro_{1}, \dots, \macvaro_{\macnumvar}\}$. Any place
  $\macvaro_{\macvaridx}$ can benefit itself, the place $\macglpl$,
  the set of places
  $\{\macvaro[\macvaridx]_{\macconstridx}^{\macdomidx} \mid 1 \le
  \macdomidx \le \macdom, \macvaro_{ \macvaridx} \text{ occurs in }
  \macconstro_{ \macconstridx}\}$ and at most $\macdeg$ places among
  $\{\macconstro_{1}, \dots, \macconstro_{\macnumconstr}\}$.  This
  adds up to at most $2+ \macdeg(\macdom+1)$.

  Suppose the given \maccsp{} instance is satisfiable. For each
  variable $\macvaro_{\macvaridx}$, if $\macdomidx$ is the domain
  value assigned to $\macvaro_{\macvaridx}$ by the satisfying
  assignment, fire the transition
  $\mactrans_{\macvaridx}^{\macdomidx}$ shown in
  \figref{fig:BDVarSet}. Since the satisfying assignment satisfies all
  the constraints, the transitions shown in \figref{fig:BDConCheck}
  can be fired to get a token into each of the places
  $\macconstro_{1}, \dots, \macconstro_{\macnumconstr}$. Then the
  transition shown in \figref{fig:BDGlCheck} can be fired to get a
  token in the place $\macglpl$. Any tokens remaining in places
  $\macvaro[\macvaridx]_{*}^{\macdomidx}$ can be removed by firing
  transitions $\mactrans[\macvaridx]_{*}^{\macdomidx}$ shown in
  \figref{fig:BDVarSet}. Now, the token in the place $\macglpl$ is the
  only token in the entire net and this is the final marking required
  to be reached.

  Suppose the required final marking is reachable in the constructed
  $1$-safe net. Consider any firing sequence reaching the required
  final marking. Since the final marking needs a token in the place
  $\macglpl$ and the only transition that can add token to $\macglpl$
  is the one shown in \figref{fig:BDGlCheck}, the firing sequence
  fires this transition. For this transition to be enabled, a token
  needs to be present in each of the places $\macconstro_{1}, \dots,
  \macconstro_{\macnumconstr}$. These tokens can only be added by
  firing transitions shown in \figref{fig:BDConCheck}. To fire these
  transitions, tokens needs to be present in the places
  $\macvaro[\macvaridx]_{*}^{\macdomidx}$. To generate these tokens,
  the firing sequence would have to fire some transition
  $\mactrans_{\macvaridx}^{\macdomidx}$ for each $\macvaridx$ between
  $1$ and $\macnumvar$. Consider the assignment that assigns domain
  value $\macdomidx$ to $\macvaro_{\macvaridx}$ iff the firing
  sequence fired $\mactrans_{\macvaridx}^{\macdomidx}$. By
  construction, this assignment satisfies all constraints.  \qed
\end{proof}
Since the $1$-safe net described above can be constructed in time
polynomial in the size of the given \maccsp{} instance,
\lemref{lem:CSP1SafeNet} shows that this reduction is a parameterized
reduction from \maccsp{} (with $\macdom$ and $\macdeg$ as parameters)
to reachability in $1$-safe nets (with benefit depth as the
parameter).  In the above reduction, it is enough to check if in the
constructed $1$-safe net, we can reach a marking that has a token at
the place $\macglpl$. This can be expressed as reachability,
coverability etc.  This proves \thmref{thm:BD1SafeReachHard}.

\end{fver}

\section{Vertex cover and model checking $1$-safe Petri nets}
In this section, we will show that with the vertex cover number of the
flow graph of the given $1$-safe Petri net and the size of the given
LTL/MSO formula as parameters, checking whether the given net is a
model of the given formula is \macfpt{}. With vertex cover number as
the only parameter, we cannot hope to get this kind of tractability:
\begin{proposition}%[*]
  \label{prop:VCOnlyVCNPHard}
  Model checking LTL (and hence MSO) formulas on $1$-safe Petri nets
  whose flow graph has constant vertex cover number is
  \macconp{}-hard.
\end{proposition}
\begin{fver}
\begin{proof}
We give a reduction from the complement of propositional logic
satisfiability problem. Let $\macpcnff$ be a propositional formula
over variables $\macpvo_{1}, \dots, \macpvo_{\macnumpv}$.  Consider
the $1$-safe net $\macnet_{\macpcnff}$ shown in
\figref{fig:VCOnlyVCNPHard}.
\begin{figure}[!htp]
  \begin{center}
    \includegraphics{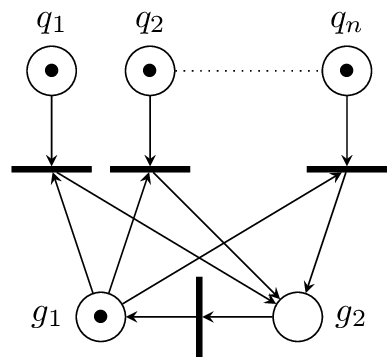}
  \end{center}
  \caption{The net $\macnet_{\macpcnff}$ associated with a
  propositional formula $\macpcnff$}
  \label{fig:VCOnlyVCNPHard}
\end{figure}
The initial marking consists of $0$ tokens in $\macglpl_{2}$ and $1$
token each in all other places. The flow graph of
$\macnet_{\macpcnff}$ has a vertex cover of size $2$ ($\{\macglpl_{1},
\macglpl_{2}\}$). Every marking $\macmark$ reachable in $\macnet_{
\macpcnff}$ defines an assignment to the variables used in
$\macpcnff$: $\macpvo_{\macpvidx} = \top$ iff $\macmark(
\macpvo_{\macpvidx}) = 1$. Every assignment can be represented by some
reachable marking in this way. We claim that $\macpcnff$ is not
satisfiable iff $\macnet_{\macpcnff}$ is a model of the LTL formula $\lnot
(\top~Until~\macpcnff)$. If $\macpcnff$ is not satisfiable, then none
of the markings reachable in $\macnet_{ \macpcnff}$ satisfies
$\macpcnff$. Hence, $\macnet_{ \macpcnff}$ is a model of the LTL
formula $\lnot (\top~Until~\macpcnff)$. On the other hand, if
$\macnet_{ \macpcnff}$ is a model of $\lnot (\top~Until~\macpcnff)$,
then none of the markings reachable in $\macnet_{ \macpcnff}$ satisfies
$\macpcnff$. Hence, $\macpcnff$ is not satisfiable.\qed
\end{proof}
\end{fver}

Since a run of a $1$-safe net $\macnet$ with set of places
$\macplaces$ is a sequence of subsets of $\macplaces$, we can think of
such sequences as strings over the alphabet $\macpset{\macplaces}$
(the power set of $\macplaces$). It is known \cite{Buchi,Vardi96} that
with any LTL or MSO formula $\macltlfo$, we can associate a finite
state automaton $\macfsa_{\macltlfo}$ over the alphabet
$\macpset{\macplaces}$ accepting the set of finite strings which are
its models, as well as a finite state B\"{u}chi automaton
$\macfsba_{\macltlfo}$ accepting the set of infinite string models.

\Figref{fig:MSExample} shows the schematic of a simple manufacturing
system modelled as a $1$-safe Petri net. Starting from
$\macplaceo_{1}$, it picks up one unit of a raw material $\macrmo$ and
goes to $\macplaceo_{2}$, then picks up raw material $\macrmt$, then
$\macrmh$. Transition $\mactranso_{1}$ does some processing and then
the system starts from $\macplaceo_{1}$ again. Suppose we want to
make sure that whenever the system picks up a unit of raw material
$\macrmt$, it is processed immediately. In other words, whenever the
system stops at a marking where no transitions are enabled, there
should not be a token in $\macplaceo_{3}$. This can be checked by
verifying that all finite maximal runs satisfy the formula
$\forall \macfovo ( (\forall \macfovt \quad \macfovt \le \macfovo)
\Rightarrow \lnot \macplaceo_{3}(\macfovo))$. The satisfaction of this
formula depends only on the number of units of raw materials $\macrmo,
\macrmt$ and $\macrmh$ at the beginning, i.e., the number of tokens at
the initial marking. The naive approach of constructing the whole
reachability graph results in an exponentially large state space, due
to the different orders in which the raw materials of each type can be
drawn. If we want to reason about only the central system (which is
the vertex cover $\{\macplaceo_{1}, \macplaceo_{2}, \macplaceo_{3},
\macplaceo_{4}\}$ in the above system), it turns out that we can
ignore the order and express the requirements on the numbers by
integer linear constraints.
\begin{figure}[!tp]
  \begin{center}
    \includegraphics{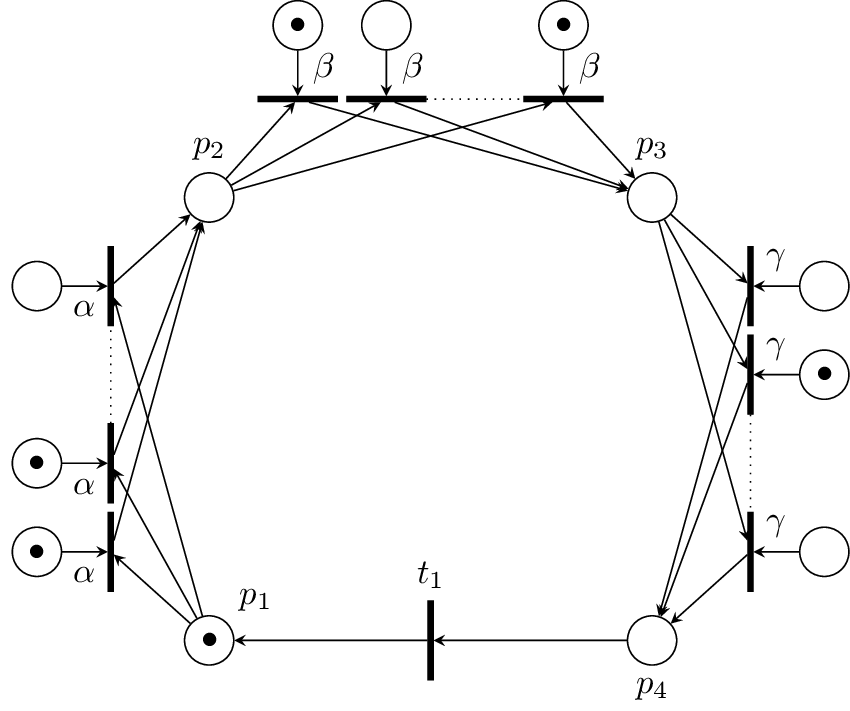}
  \end{center}
  \caption{An example of a system with small vertex cover}
  \label{fig:MSExample}
\end{figure}

Suppose $\macvertcov$ is a vertex cover for $\macgraph(\macnet)$. We
use the fact that if $\macverto_{1},\macverto_{2}\notin \macvertcov$
are two vertices not in $\macvertcov$ that have the same set of
neighbours, $\macverto_{1}$ and $\macverto_{2}$ have similar
properties. This has been used to obtain \macfpt{} algorithms for many
hard problems (e.g. \cite{FLMRS08}). The following definitions
formalize this.
\begin{definition}
Let $\macvertcov$ be a vertex cover of $\macgraph(\macnet)$. The
\concept{($\macvertcov$-) neighbourhood} of a transition $\mactranso$
is the ordered pair $( \macipls{ \mactranso} \cap \macvertcov,
\macopls{ \mactranso} \cap \macvertcov)$.  We denote by
$\macnumtranstype$ the number of different
$\macvertcov$-neighbourhoods.
\end{definition}
%In \figref{fig:VCExample}, transitions
%$\mactranso_{1}$ and $\mactranso_{5}$ are of the same type.
%\begin{figure}[!htp]
%  \begin{center}
%    \includegraphics{VCExample}
%  \end{center}
%  \caption{A Petri net with vertex cover $\left\{
%  \macplaceo_{1},\dots,\macplaceo_{4} \right\}$}
%  \label{fig:VCExample}
%\end{figure}
\begin{definition}
  \label{def:placeVariety}
  Suppose $\macnet$ is a Petri net with $\macnumtranstype$ neighbourhoods
  for vertex cover $\macvertcov$, and $\macplaceo\notin \macvertcov$. 
  The \concept{($\macvertcov$-) interface}
  $\macvar[\macplaceo]$ of $\macplaceo$ is defined as the
  function
  $\macvar[\macplaceo]: \{1,\dots,\macnumtranstype\} \to \macpset{\{-1,1\}}$, 
  where for every $\mactypeidx$ between $1$ and
  $\macnumtranstype$ and every $\macarcw \in \{1,-1\}$, there is a transition $\mactranso_{\mactypeidx}$
  of VC-neighbourhood $\mactypeidx$ such that $\macarcw =
  -\macpre(\macplaceo,\mactranso_{\mactypeidx})+\macpost(\macplaceo,
  \mactranso_{\mactypeidx})$ iff $\macarcw \in \macvar[\macplaceo]
  (\mactypeidx)$.  
\end{definition}
In the net in \figref{fig:MSExample} with
$\macvertcov=\{\macplaceo_{1}, \macplaceo_{2}, \macplaceo_{3},
\macplaceo_{4}\}$, all transitions labelled $\macrmo$ have the same
VC-neighbourhood and all the corresponding places have the same
VC-interface.  Since there can be $2\macvcnum$ arcs between a
transition and places in \macvertcov\/ if $|\macvertcov|=\macvcnum$,
there can be at most $2^{2\macvcnum}$ different VC-neighbourhoods of
transitions.  There are at most $4^{2^{2\macvcnum}}$ VC-interfaces.
The set of interfaces is denoted by $\macplvars$. 

%Place $\macplaceo$ can have one incoming or one outgoing arc
%with each transition of the net. It cannot have both an incoming and an
%outgoing arc since in that case, $\macplaceo$ would have a self loop
%and would be in $\macvertcov$. If $\macplaceo'$ is another place not
%in $\macvertcov$, then no transition can have arcs to both
%$\macplaceo$ and $\macplaceo'$, since otherwise, there would have been
%an edge between $\macplaceo$ and $\macplaceo'$ in $\macgraph(\macnet)$
%and one of the places $\macplaceo$ and $\macplaceo'$ would have been
%in $\macvertcov$.  Hence, places not in $\macvertcov$ cannot interact
%with each other directly. Places not in $\macvertcov$ can only
%interact with places in $\macvertcov$ through transitions and there
%are at most $\macnumtranstype$ VC-neighbourhoods of transitions. 
%Suppose $\macplaceo$ and $\macplaceo'$ share the same interface.
%(for example, $\macplaceo_{5}$ and $\macplaceo_{6}$ in
%\figref{fig:VCExample}).  At most one among
%$\macpre(\macplaceo,\mactranso_{\mactypeidx})$ and
%$\macpost(\macplaceo,\mactranso_{\mactypeidx})$ will be non-zero.

\begin{proposition}
  \label{prop:PlSameVarTransEnable}
  Let $\macnet$ be a $1$-safe net with $\macvertcov$ being a vertex
  cover of $\macgraph(\macnet)$. Let $\macplace_{1}, \macplace_{2},
  \dots, \macplace_{\macplaceidx}$ be places not in the vertex cover,
  all with the same interface. Let $\macmark$ be some marking reachable
  from the initial marking of $\macnet$. If
  $\macmark(\macplace_{\macplaceidxt})=1$ for some $\macplaceidxt$ between
  $1$ and $\macplaceidx$, then $\macmark$ does not enable any
  transition that adds tokens to any of the places
  $\macplace_{1}, \dots, \macplace_{\macplaceidx}$.
\end{proposition}
\begin{proof}
  Suppose there is a transition $\mactrans$ enabled at $\macmark$ that
  adds a token to $\macplace_{\macplaceidxt'}$ for some
  $\macplaceidxt'$ between $1$ and $\macplaceidx$. Then there is a
  transition $\mactrans'$ with the same neighbourhood as $\mactrans$
  (and hence enabled at $\macmark$ too) that can add a token to
  $\macplace_{\macplaceidxt}$. Firing $\mactrans'$ from $\macmark$
  will create $2$ tokens at $\macplace_{\macplaceidxt}$, contradicting
  the fact that $\macnet$ is $1$-safe.
\qed
\end{proof}
If the initial marking has tokens in many places with the same
interface, then no transition can add tokens to any of those places
until all the tokens in all those places are removed. Once all tokens
are removed, one of the places can receive one token after which, no
place can receive tokens until this one is removed. All these places
have the same interface. Thus, a set of places with the same interface
can be thought of as an initial storehouse of tokens, after depleting
which it can be thought of as a single place.  However, a formula in
our logic can reason about individual places, so we still need to keep
track of individual places that occur in the formula.
\begin{proposition}%[*]
  \label{prop:LTLInvOtherPl}
  Let $\macnet$ be a $1$-safe net and $\macltlfo$ be an MSO formula.
  Let $\macplaces_{\macltlfo} \subseteq \macplaces$ be the subset of
  places that occur in $\macltlfo$. Let $\macmarks = \macmark_{0}
  \macmark_{1} \cdots$ and $\macmarks'= \macmark_{0}' \macmark_{1}'
  \cdots$ be two finite or infinite runs of $\macnet$ such that
  for all positions $\macseqidx$ of $\macmarks$ and for all $\macplace
  \in \macplaces_{\macltlfo}$, $\macmark_{\macseqidx}(\macplace) =
  \macmark_{\macseqidx}'(\macplace)$. For any assignment %macah
  $s$, we have $\macmarks, s \models
  \macltlfo$ iff $\macmarks', s \models \macltlfo$.
\end{proposition}
\begin{fver}
  \begin{proof}
    By a straightforward induction on the structure of $\macltlfo$.
  \qed
  \end{proof}
\end{fver}

Let $\macnet$ be a $1$-safe net such that $\macgraph(\macnet)$ has a
vertex cover $\macvertcov$ of size $\macvcnum$. Suppose $\macltlfo$ is
a formula and we have to check if $\macnet$ satisfies $\macltlfo$.
For each interface $\macplvar$, let $\macplaces_{\macplvar} \subseteq
\macplaces$ be the places not in $\macvertcov$ with interface
$\macplvar$. If $\macplaces_{\macplvar} \setminus
\macplaces_{\macltlfo} \ne \emptyset$ (i.e., if there are places in
$\macplaces_{\macplvar}$ that are not in $\macltlfo$), designate one
of the places in $\macplaces_{\macplvar} \setminus
\macplaces_{\macltlfo}$ as $\macplace_{\macplvar}$. Define the
set of special places $\macsplaces = \macvertcov \cup
\macplaces_{\macltlfo} \cup \{\macplace_{\macplvar} \in
\macplaces_{\macplvar} \setminus \macplaces_{\macltlfo} \mid \macplvar
\text{ is an interface and } \macplaces_{\macplvar} \setminus
\macplaces_{\macltlfo} \ne \emptyset\}$. Note that $|\macsplaces| \le
\macvcnum + |\macltlfo| + 4^{2^{2\macvcnum}}$. Since this number is a
function of the parameters of the input instance, we will treat it as
a parameter.

We need a structure that keeps track of changes in places belonging to
$\macsplaces$, avoiding a construction involving all reachable
markings. This can be done by a finite state machine whose states are
subsets of $\macsplaces$. Transitions of the Petri net that only
affect places in $\macsplaces$ can be simulated by the finite state
machine with its usual transitions. To simulate transitions of the net
that affect places outside $\macsplaces$, we need to impose some
conditions on the number of times transitions of the finite state
machine can be used. The following definition formalizes this.  For a
marking $\macmark$ of $\macnet$, let $\macmark\macrestr{\macsplaces} =
\{\macplace\in \macsplaces \mid \macmark(\macplace) = 1\}$.
\begin{definition}
  \label{def:ConstrAutNet}
  Given a $1$-safe net $\macnet$ with initial marking
  $\macmark_{0}$ and $\macsplaces$ defined from $\macltlfo$ as above, 
  the \concept{edge constrained automaton} 
  $\macfsa_{\macnet} = (\macstates_{\macnet}, \macalph,
  \mactransrel_{\macnet}, \macedgeconstr, \macfinss_{\macnet})$ is a structure
  defined as follows.  $\macstates_{\macnet} = \macpset{\macsplaces}$
  and $\macalph = \macplvars \cup \{\bot\}$ 
  (recall that $\macplvars$ is the set of interfaces in
  $\macnet$).  The transition relation $\mactransrel \subseteq
  \macstates_{\macnet} \times \macalph \times \macstates_{\macnet}$ is
  such that for all $\macplaces_{1}, \macplaces_{2}
  \subseteq \macsplaces$ and $\macplvar \in \macplvars \cup \{\bot\}$,
  $(\macplaces_{1}, \macplvar, \macplaces_{2}) \in
  \mactransrel$ iff there are markings $\macmark_{1}, \macmark_{2}$
  and a transition $\mactrans$ of $\macnet$ such that
  \begin{itemize}
    \item $\macmark_{1}\macrestr{\macsplaces} = \macplaces_{1}$,
      $\macmark_{2} \macrestr{\macsplaces} = \macplaces_{2}$ and
      $\macmark_{1} \macStep{\mactrans} \macmark_{2}$,
    \item $\mactrans$ removes a token from a place $\macplace \in
      \macplaces_{\macplvar} \setminus \macsplaces$ of interface
      $\macplvar$ if $\macplvar \in \macplvars$ and
    \item $\mactrans$ does not have any of its input or output places
      in $\macplaces \setminus \macsplaces$ if $\macplvar = \bot$.
  \end{itemize}
  The \concept{edge constraint} $\macedgeconstr:\macplvars \to \macNat$ 
  is given by
  $\macedgeconstr(\macplvar) = |\{\macplace \in
  \macplaces_{\macplvar}\setminus \macsplaces \mid
  \macmark_{0}(\macplace) = 1\}|$. A subset 
  $\macplaces_{1} \subseteq \macsplaces$ is in 
  $\macfinss_{\macnet}$ iff for every marking $\macmark$ with
  $\macmark\macrestr{\macsplaces} = \macplaces_{1}$, the only
  transitions enabled at $\macmark$ remove tokens from
  some place not in $\macsplaces$.
\end{definition}
Intuitively, the edge constraint $\macedgeconstr$ defines an upper
bound on the number of times those transitions can be used that reduce
tokens from places not in $\macsplaces$.
\begin{definition}
  \label{def:ConstrAutAcc}
  Let $\macfsa_{\macnet}$ be an edge constrained automaton as in
  \defref{def:ConstrAutNet} and let $\macmarks = \macplaces_{0}
  \macplaces_{1} \cdots$ be a finite or infinite word over
  $\macpset{\macsplaces}$. Then $\macmarks$ is a valid run of
  $\macfsa_{\macnet}$ iff for every position $\macseqidx \ge 1$ of
  $\macmarks$, we can associate an element $\macplvar_{\macseqidx} \in
  \macalph$ such that
  \begin{itemize}
    \item for every position $\macseqidx \ge 1$ of $\macmarks$,
      $(\macplaces_{\macseqidx-1}, \macplvar_{\macseqidx},
      \macplaces_{\macseqidx}) \in \mactransrel$ and
    \item for every $\macplvar \in \macplvars$, $|\{\macseqidx \ge
      1 \mid \macplvar_{\macseqidx}=\macplvar\}| \le
      \macedgeconstr(\macplvar)$.
    \item if $\macmarks$ is finite and
      $\macplaces_{\macseqidx}$ is the last element of $\macmarks$, then
      $\macplaces_{\macseqidx} \in \macfinss_{\macnet}$ and
      for every interface $\macplvar \in \macplvars$ and marking
      $\macmark_{\macseqidx}
      \macrestr{\macsplaces} = \macplaces_{\macseqidx}$ enabling some
      transition that removes tokens from some place in
      $\macplaces_{\macplvar} \setminus \macsplaces$,
      $|\{\macseqidx \ge
      1 \mid \macplvar_{\macseqidx}=\macplvar\}| =
      \macedgeconstr(\macplvar)$.
  \end{itemize}
\end{definition}

Next we have a run construction lemma.
\begin{lemma}%[*]
  \label{lem:NetFirSeqConstrAutAcc}
  Let $\macnet$ be a $1$-safe net with initial marking $\macmark_{0}$,
  $\macltlfo$ be a formula and $\macfsa_{\macnet}$ be as in
  \defref{def:ConstrAutNet}. For every infinite (maximal finite) run
  $\macmarks = \macmark_{0} \macmark_{1} \cdots$ of $\macnet$, there
  exists an infinite (finite) run $\macmarks' =\macmark_{0}'
  \macmark_{1}' \cdots$ such that the word $(\macmark_{0}'
  \macrestr{\macsplaces}) (\macmark_{1}'\macrestr{\macsplaces}) \cdots
  $ is a valid run of $\macfsa_{\macnet}$ and for every position
  $\macseqidx$ of $\macmarks$, $\macmark_{\macseqidx}'
  \macrestr{\macplaces_{\macltlfo}} = \macmark_{\macseqidx}
  \macrestr{\macplaces_{\macltlfo}}$.  If an infinite (finite) word
  $\macmarks = \macplaces_{0} \macplaces_{1} \cdots$ over
  $\macpset{\macsplaces}$ is a valid run of $\macfsa_{\macnet}$ and
  $\macplaces_{0} = \macmark_{0} \macrestr{\macsplaces}$, then there
  is an infinite (finite maximal) run $\macmark_{0} \macmark_{1}
  \cdots$ of $\macnet$ such that $\macmark_{\macseqidx}
  \macrestr{\macsplaces} = \macplaces_{\macseqidx}$ for all positions
  $\macseqidx$ of $\macmarks$.
\end{lemma}
\begin{fver}
  \begin{proof}
      Let $\macmarks = \macmark_{0} \macmark_{1} \cdots$ be an infinite or
  a maximal finite run of $\macnet$. For every interface $\macplvar
  \in \macplvars$, perform the following steps: if for some marking
  $\macmark$ in the above run, $\{ \macplace\in
  \macplaces_{\macplvar}\mid \macmark( \macplace) = 1\} = \emptyset$,
  let $\macmark_{\macplvar}$ be the first such marking. By
  \propref{prop:PlSameVarTransEnable}, no transition occurring before
  $\macmark_{\macplvar}$ will add any token to any place in
  $\macplaces_{\macplvar}$. If there is any transition occurring after
  $\macmark_{\macplvar}$ that adds/removes tokens from
  $\macplaces_{\macplvar}\setminus \macsplaces$, replace it with
  another transition with the same neighbourhood that adds/removes
  tokens from $\macplace_{\macplvar}$. By
  \propref{prop:PlSameVarTransEnable}, such a replacement will not
  affect any place in $\macplaces_{\macltlfo}$ and the new sequence of
  transitions is still enabled at $\macmark_{0}$. After performing
  this process for every interface $\macplvar \in \macplvars$, let the
  new run be $\macmarks' = \macmark_{0}' \macmark_{1}' \cdots$. By
  construction, we have $\macmark_{\macseqidx}'
  \macrestr{\macplaces_{\macltlfo}} = \macmark_{\macseqidx}
  \macrestr{\macplaces_{\macltlfo}}$ for all positions $\macseqidx \ge
  0$ of $\macmarks$. If $\macmarks$ is a maximal finite run, so is
  $\macmarks'$.

  Now we will prove that the word $(\macmark_{0}' \macrestr{\macsplaces})
  (\macmark_{1}' \macrestr{\macsplaces}) \cdots$
  is a valid run of $\macfsa_{\macnet}$. Suppose the sequence of
  transitions producing the run $\macmarks'$ is
  $\macmark_{0}' \macStep{\mactrans_{1}} \macmark_{1}'
  \macStep{ \mactrans_{2}} \macmark_{2}' \cdots$. For each position
  $\macseqidx \ge 1$ of this run, define $\macplvar_{\macseqidx} \in
  \macalph$ as follows:
  \begin{itemize}
    \item if $\mactrans_{\macseqidx}$ has all its input and output
      places among places $\macsplaces$, let $\macplvar_{\macseqidx} =
      \bot$.
    \item if $\mactrans_{\macseqidx}$ removes a token from some place
      in $\macplaces_{\macplvar} \setminus \macsplaces$ for some
      interface $\macplvar$, let $\macplvar_{\macseqidx} = \macplvar$.
      Due to the way $\macmarks'$ is constructed, this kind of
      transition can only occur before the position of
      $\macmark_{\macplvar}$ and the number of such occurrences is at
      most $|\{\macplace \in \macplaces_{\macplvar} \setminus
      \macsplaces \mid \macmark_{0}(\macplace) = 1\}| =
      \macedgeconstr(\macplvar)$.
  \end{itemize}
  Due to the way $\macmarks'$ is constructed, there will not be any
  transition that adds tokens to any place in $\macplaces_{\macplvar}
  \setminus \macsplaces$ for any interface $\macplvar$. By definition,
  it is clear that for every position $\macseqidx \ge 1$ of
  $\macmarks'$, $(\macmark_{\macseqidx-1}' \macrestr{\macsplaces},
  \macplvar_{\macseqidx}, \macmark_{\macseqidx}'
  \macrestr{\macsplaces}) \in \mactransrel_{\macnet}$. In addition,
  for every interface $\macplvar \in \macplvars$, we have $|\{
  \macseqidx \ge 1 \mid \macplvar_{\macseqidx} = \macplvar\}| \le
  \macedgeconstr(\macplvar)$. Hence, the word 
  $(\macmark_{0}' \macrestr{\macsplaces})
   (\macmark_{1}' \macrestr{\macsplaces}) \cdots$
  is a valid run of $\macfsa_{\macnet}$ if the word is infinite. If
  $\macmarks'$ is finite, suppose $\macmark_{\macfslidx}'$ is the
  last marking of the sequence $\macmarks'$. Suppose for some variety
  $\macplvar \in \macplvars$, there is some marking
  $\macmark$ such that $\macmark
  \macrestr{\macsplaces} = \macmark_{\macfslidx}'
  \macrestr{\macsplaces}$ and $\macmark$ enables some
  transition $\mactrans$ that removes tokens from some place in
  $\macplaces_{\macplvar} \setminus \macsplaces$. Since
  $\macmark_{\macfslidx}'$ does not enable any transition, all
  transitions (including $\mactrans$) removing tokens from some place
  in $\macplaces_{\macplvar} \setminus \macsplaces$ are disabled in
  $\macmark_{\macfslidx}'$. This means that every place in
  $\macplaces_{\macplvar} \setminus \macsplaces$ that had a token in
  $\macmark_{0}$ has lost its token in $\macmark_{\macfslidx}'$.
  Since such loss of tokens can only happen by firing transitions that
  remove tokens from places in $\macplaces_{\macplvar} \setminus
  \macsplaces$, we have $|\{\macseqidx \ge 1 \mid
  \macplvar_{\macseqidx} = \macplvar\}|=\macedgeconstr(\macplvar)$.
  Hence, to prove that $(\macmark_{0}' \macrestr{\macsplaces})
  (\macmark_{1}' \macrestr{\macsplaces}) \cdots (\macmark_{
  \macfslidx}' \macrestr{\macsplaces})$ is a valid run of
  $\macfsa_{\macnet}$, it is left to show that $\macmark_{\macfslidx}'
  \in \macfinss_{\macnet}$. To see that this is true, observe that if
  some marking $\macmark$ with $\macmark \macrestr{\macsplaces} =
  \macmark_{\macfslidx}'$ enables a transition that does not remove
  any token from $\macplaces \setminus \macsplaces$, then so does
  $\macmark_{\macfslidx}'$, a contradiction.

  Next, suppose $\macmarks = \macplaces_{0} \macplaces_{1} \cdots$ is
  an infinite or finite word that is a valid run
  of $\macfsa_{\macnet}$ such that $\macplaces_{0} = \macmark_{0}
  \macrestr{\macsplaces}$. For every position $\macseqidx \ge 1$ of
  $\macmarks$, there are $\macplvar_{\macseqidx} \in \macplvars \cup
  \{\bot\}$, transition $\mactrans_{\macseqidx}'$ and markings
  $\macmark_{\macseqidx-1}'$ and $\macmark_{\macseqidx}'$ such that
  $(\macplaces_{\macseqidx-1}, 
  \macplvar_{\macseqidx}, \macplaces_{\macseqidx}) \in
  \mactransrel_{\macnet}$, $\macmark_{\macseqidx-1}'
  \macStep{\mactrans_{\macseqidx}'} \macmark_{\macseqidx}'$,
  $\macmark_{\macseqidx-1}' \macrestr{ \macsplaces} =
  \macplaces_{\macseqidx-1}$ and $\macmark_{\macseqidx}' \macrestr{
  \macsplaces} = \macplaces_{\macseqidx}$. Define transitions
  $\mactrans_{i}$ as follows:
  \begin{itemize}
    \item If $\macplvar_{\macseqidx} = \bot$, transition
      $\mactrans_{\macseqidx}'$ has all its input and output places in
      $\macsplaces$. Let $\mactrans_{\macseqidx} =
      \mactrans_{\macseqidx}'$.
    \item If $\macplvar_{\macseqidx} = \macplvar \in \macplvars$,
      transition $\mactrans_{\macseqidx}'$ removes a token from some
      place in $\macplaces_{\macplvar} \setminus \macsplaces$. Let
      $\mactrans'$ be a transition of the same neighbourhood as
      $\mactrans_{\macseqidx}'$ that removes a token from some place
      $\macplace_{\macseqidx} \in \{\macplace \in
      \macplaces_{\macplvar} \setminus \macsplaces \mid
      \macmark_{0}(\macplace) =1\}$ such that no transition among
      $\mactrans_{1}, \dots, \mactrans_{\macseqidx-1}$ removes tokens
      from $\macplace_{\macseqidx}$. This is possible since, due to
      the validity of $\macmarks$ in $\macfsa_{\macnet}$, $|\{
      \macseqidx' \ge 1 \mid \macplvar_{\macseqidx'} = \macplvar\}|
      \le |\{\macplace \in \macplaces_{\macplvar} \setminus
      \macsplaces \mid \macmark_{0}(\macplace) =1\}| = \macedgeconstr
      ( \macplvar)$. Let $\mactrans_{\macseqidx} = \mactrans'$.
  \end{itemize}
  We will now prove by induction on $\macseqidx$ that there are
  markings $\macmark_{0}, \macmark_{1}, \dots$ such that
  $\macmark_{0} \macStep{\mactrans_{1}} \macmark_{1}
  \macStep{\mactrans_{2}} \cdots \macStep{\mactrans_{\macseqidx}}
  \macmark_{\macseqidx}$ and $\macmark_{\macseqidx}
  \macrestr{\macsplaces} = \macplaces_{\macseqidx}$ for every position
  $\macseqidx$ of $\macmarks$.

  \emph{Base case $\macseqidx = 1$:} If $\macplvar_{1} = \bot$, the
  fact that $\macmark_{0} \macrestr{\macsplaces} = \macplaces_{0}$,
  $\macmark_{0}' \macStep{\mactrans_{1}} \macmark_{1}'$ and that
  $\mactrans_{1}$ has all its input and output places in $\macsplaces$
  implies that $\macmark_{0} \macStep{\mactrans_{1}} \macmark_{1}$ for
  some $\macmark_{1}$ such that $\macmark_{1} \macrestr{\macsplaces} =
  \macplaces_{1}$. If $\macplvar_{\macseqidx} = \macplvar \in
  \macplvars$, then $\mactrans_{1}$ removes a token from some place
  $\macplace_{1} \in \macplaces_{\macplvar} \setminus \macsplaces$.
  Again the fact that $\macmark_{0} \macrestr{\macsplaces} =
  \macplaces_{0}$ and $\macmark_{0}' \macStep{\mactrans_{1}'}
  \macmark_{1}'$ implies that $\macmark_{0} \macStep{\mactrans_{1}}
  \macmark_{1}$ for some $\macmark_{1}$ such that $\macmark_{1}
  \macrestr{\macsplaces} = \macplaces_{1}$.

  \emph{Induction step:} If $\macplvar_{\macseqidx+1} = \bot$, the
  fact that $\macmark_{\macseqidx} \macrestr{\macsplaces} =
  \macplaces_{\macseqidx}$, $\macmark_{\macseqidx}'
  \macStep{\mactrans_{\macseqidx+1}} \macmark_{\macseqidx+1}'$ and
  that $\mactrans_{\macseqidx+1}$ has all its input and output places
  in $\macsplaces$ implies that $\macmark_{\macseqidx}
  \macStep{\mactrans_{\macseqidx+1}} \macmark_{\macseqidx+1}$ for some
  $\macmark_{\macseqidx+1}$ such that $\macmark_{\macseqidx+1}
  \macrestr{\macsplaces} = \macplaces_{\macseqidx+1}$. If
  $\macplvar_{\macseqidx+1} = \macplvar \in \macplvars$, then
  $\mactrans_{\macseqidx+1}$ removes a token from some place
  $\macplace_{\macseqidx+1} \in \macplaces_{\macplvar} \setminus
  \macsplaces$. Again the fact that $\macmark_{\macseqidx}
  \macrestr{\macsplaces} = \macplaces_{\macseqidx}$ and
  $\macmark_{\macseqidx}' \macStep{\mactrans_{\macseqidx+1}'}
  \macmark_{\macseqidx+1}'$ implies that $\macmark_{\macseqidx}
  \macStep{\mactrans_{\macseqidx+1}} \macmark_{\macseqidx+1}$ for some
  $\macmark_{\macseqidx+1}$ such that $\macmark_{\macseqidx+1}
  \macrestr{\macsplaces} = \macplaces_{\macseqidx+1}$.

  If $\macmarks$ is a finite word, we have to prove that the run
  constructed above is a maximal finite run. Let
  $\macmark_{\macfslidx}$ be the last marking in the sequence
  constructed above. We will prove that $\macmark_{\macfslidx}$ does
  not enable any transition. Suppose some transition $\mactrans$ is
  enabled at $\macmark_{\macfslidx}$.  Since $\macmark_{\macfslidx}
  \macrestr{ \macsplaces} \in \macfinss_{\macnet}$, $\mactrans$
  removes a token from some place in $\macplaces_{\macplvar} \setminus
  \macsplaces$ for some variety $\macplvar$. Since $|\{\macseqidx \ge
  1 \mid \macplvar_{\macseqidx} = \macplvar\}| =
  \macedgeconstr(\macplvar)$, there are $\macedgeconstr(\macplvar)$
  transition occurrences among $\mactrans_{1}, \dots,
  \mactrans_{\macfslidx}$ that each remove a token from some place in
  $\macplaces_{\macplvar} \setminus \macsplaces$.  Since there were
  exactly $\macedgeconstr(\macplvar)$ places in
  $\macplaces_{\macplvar} \setminus \macsplaces$ that had a token in
  $\macmark_{0}$ and no other transition adds any token to any place
  in $\macplaces_{\macplvar} \setminus \macsplaces$, $\mactrans$ can
  not be enabled at $\macmark_{\macfslidx}$. Hence, no transition is
  enabled at $\macmark_{\macfslidx}$.
\qed
  \end{proof}
\end{fver}

%\begin{lemma}
%  \label{lem:NetFirSeqConstrAutAccFin}
%  Let $\macnet$ be a $1$-safe net with initial marking $\macmark_{0}$,
%  $\macltlfo$ be a formula and $\macfsa_{\macnet}$ be as in
%  \defref{def:ConstrAutNet}. For a maximal finite run  $\macmarks =
%  \macmark_{0} \macmark_{1} \cdots$ of $\macnet$, there is a finite
%  run $\macmarks' =\macmark_{0}' \macmark_{1}' \cdots$ such that the
%  word $(\macmark_{0}' \macrestr{\macsplaces})
%  (\macmark_{1}'\macrestr{\macsplaces}) \cdots $ is a valid run of
%  $\macfsa_{\macnet}$ and for every position $\macseqidx$ of
%  $\macmarks$, $\macmark_{\macseqidx}'
%  \macrestr{\macplaces_{\macltlfo}} = \macmark_{\macseqidx}
%  \macrestr{\macplaces_{\macltlfo}}$.  If a finite word $\macmarks =
%  \macplaces_{0} \macplaces_{1} \cdots$ is a valid run of
%  $\macfsa_{\macnet}$ and $\macplaces_{0} = \macmark_{0}
%  \macrestr{\macsplaces}$, then there is a finite maximal run
%  $\macmark_{0} \macmark_{1} \cdots$ of $\macnet$ with
%  $\macmark_{\macseqidx} \macrestr{\macsplaces} =
%  \macplaces_{\macseqidx}$ for all positions $\macseqidx$ of
%  $\macmarks$.
%\end{lemma}

\Lemref{lem:NetFirSeqConstrAutAcc} implies that in order to check if
$\macnet$ is a model of the formula $\macltlfo$, it is enough to check
that all valid runs of $\macfsa_{\macnet}$ satisfy $\macltlfo$. This
can be done by checking that no finite valid run of
$\macfsa_{\macnet}$ is accepted by $\macfsa_{\lnot \macltlfo}$ and no
infinite valid run of $\macfsa_{\macnet}$ is accepted by
$\macfsba_{\lnot \macltlfo}$.  As usual, this needs a product
construction.  Automata $\macfsa_{\lnot \macltlfo}$ and
$\macfsba_{\lnot \macltlfo}$ run on the alphabet
$\macpset{\macplaces_{\macltlfo}}$. Let $\macstates_{\macfsa}$ and
$\macstates_{\macfsba}$ be the set of states of $\macfsa_{\lnot
\macltlfo}$ and $\macfsba_{\lnot \macltlfo}$ respectively. Then,
$\macfsa_{\lnot \macltlfo} = (\macstates_{\macfsa},
\macpset{\macplaces_{\macltlfo}}, \mactransrel_{\macfsa},
\macstates_{0 \macfsa}, \macfinss_{\macfsa})$ and $\macfsba_{\lnot
\macltlfo} = ( \macstates_{\macfsba},
\macpset{\macplaces_{\macltlfo}}, \mactransrel_{\macfsba},
\macstates_{0 \macfsba}, \macfinss_{\macfsba})$.
\begin{definition}
  \label{def:NetAutProd}
  $\macfsa_{\macnet} \times \macfsa_{\lnot \macltlfo} = (
  \macstates_{\macnet}\times \macstates_{\macfsa}, \macalph,
  \mactransrel_{\macfsa}^{\macnet},
  \{\macmark_{0}\macrestr{\macsplaces}\}\times \macstates_{0 \macfsa},
  \macfinss_{\macnet} \times \macfinss_{\macfsa}, \macedgeconstr)$,
  $\macfsa_{\macnet} \times \macfsba_{\lnot \macltlfo} = (
  \macstates_{\macnet} \times \macstates_{\macfsba}, \macalph,
  \mactransrel_{\macfsba}^{\macnet}, \{\macmark_{0}
  \macrestr{\macsplaces} \}\times \macstates_{0 \macfsba},
  \macstates_{\macnet} \times \macfinss_{\macfsba}, \macedgeconstr)$ where
  \begin{align*}
    ( (\macstateo_{1}, \macstateo_{2}), \macplvar,
    ( \macstateo_{1}', \macstateo_{2}')) \in
    \mactransrel_{\macfsa}^{\macnet} \text{ iff } (
    \macstateo_{1}, \macplvar, \macstateo_{1}') \in
    \mactransrel_{\macnet} \text{ and } (\macstateo_{2},
    \macstateo_{1} \cap \macplaces_{\macltlfo}, \macstateo_{2}') \in
    \mactransrel_{\macfsa}\\
    ( (\macstateo_{1}, \macstateo_{2}), \macplvar,
    ( \macstateo_{1}', \macstateo_{2}')) \in
    \mactransrel_{\macfsba}^{\macnet} \text{ iff } (
    \macstateo_{1}, \macplvar, \macstateo_{1}') \in
    \mactransrel_{\macnet} \text{ and } (\macstateo_{2},
    \macstateo_{1} \cap \macplaces_{\macltlfo}, \macstateo_{2}') \in
    \mactransrel_{\macfsba}
  \end{align*}
  An accepting path of $\macfsa_{\macnet} \times \macfsa_{\lnot
  \macltlfo}$ is a sequence $(\macstateo_{0},\macstateo_{0}') \macplvar_1
  (\macstateo_{1},\macstateo_{1}') \cdots \macplvar_{\macfslidx}
  (\macstateo_{\macfslidx},\macstateo_{\macfslidx}')$ 
  which is $\mactransrel_{\macfsa}^{\macnet}$-respecting:
  \begin{itemize}
    \item $(\macstateo_{0},\macstateo_{0}'),
      (\macstateo_{1},\macstateo_{1}'), \dots,
      (\macstateo_{\macfslidx},\macstateo_{\macfslidx}') \in
      \macstates_{\macnet} \times \macstates_{\macfsa}$,
    \item the word $\macplvar_{1} \cdots \macplvar_{\macfslidx} \in
      \macalph^{*}$ witnesses the validity of the run $\macstateo_{0}
      \macstateo_{1} \cdots \macstateo_{\macfslidx}$ in $\macfsa_{
      \macnet}$ (as in \defref{def:ConstrAutAcc}) and
    \item the word $(\macstateo_{0} \cap \macplaces_{\macltlfo})
      \cdots (\macstateo_{\macfslidx} \cap \macplaces_{\macltlfo})$ is
      accepted by $\macfsa_{\lnot \macltlfo}$ through the run
      $\macstateo_{0}' \macstateo_{1}' \cdots \macstateo_{\macfslidx}'
      \macstateo_{\macfinss}'$ for some $\macstateo_{\macfinss}' \in
      \macfinss_{\macfsa}$ with $( \macstateo_{\macfslidx}',
      \macstateo_{\macfslidx}\cap \macplaces_{\macltlfo},
      \macstateo_{\macfinss}' )\in \mactransrel_{\macfsa}$.
  \end{itemize}
  An accepting path of
  $\macfsa_{\macnet} \times \macfsba_{\lnot \macltlfo}$ is defined
  similarly.
\end{definition}
\begin{proposition}%[*]
  \label{prop:ModelCheckAccPath}
  A $1$-safe net $\macnet$ with initial marking $\macmark_{0}$
  is a model of a formula $\macltlfo$ iff there is no accepting path
  in $\macfsa_{\macnet} \times \macfsa_{\lnot \macltlfo}$ and
  $\macfsa_{\macnet} \times \macfsba_{\lnot \macltlfo}$.
\end{proposition}
\begin{fver}
\begin{proof}
  Suppose $\macnet$ is a model of $\macltlfo$. Hence, all maximal runs
  of $\macnet$ satisfy $\macltlfo$. We will prove that there is no
  accepting path in $\macfsa_{\macnet} \times \macfsa_{\lnot
  \macltlfo}$ and $\macfsa_{\macnet} \times \macfsba_{\lnot
  \macltlfo}$. Assume by way of contradiction that there is an
  accepting path $(\macstateo_{0},\macstateo_{0}') \macplvar_1
  (\macstateo_{1},\macstateo_{1}') \cdots \macplvar_{\macfslidx}
  (\macstateo_{\macfslidx},\macstateo_{\macfslidx}')$ in
  $\macfsa_{\macnet} \times \macfsa_{\lnot \macltlfo}$. By
  \defref{def:NetAutProd}, $\macstateo_{0} \macstateo_{1} \cdots
  \macstateo_{\macfslidx}$ is a valid run of $\macfsa_{\macnet}$. By
  \lemref{lem:NetFirSeqConstrAutAcc}, there is a finite maximal run
  $\macmark_{0} \macmark_{1} \cdots \macmark_{ \macfslidx}$ of
  $\macnet$ with $\macmark_{\macseqidx} \macrestr{\macsplaces} =
  \macstateo_{\macseqidx}$ for all positions $0 \le \macseqidx \le
  \macfslidx$. By \defref{def:NetAutProd}, $(\macstateo_{0} \cap
  \macplaces_{\macltlfo}) \cdots (\macstateo_{\macfslidx} \cap
  \macplaces_{\macltlfo})$ is accepted by $\macfsa_{\lnot \macltlfo}$
  and hence satisfies $\lnot \macltlfo$. \Propref{prop:LTLInvOtherPl}
  now implies that $\macmark_{0} \macmark_{1} \cdots \macmark_{
  \macfslidx}$ satisfies $\lnot \macltlfo$, a contradiction. The
  argument for $\macfsa_{\macnet} \times \macfsba_{\lnot \macltlfo}$
  is similar.

  Suppose $\macnet$ is not a model of $\macltlfo$. Suppose there is a
  finite maximal run $\macmark_{0} \macmark_{1} \cdots \macmark_{
  \macfslidx}$ of $\macnet$ that satisfies $\lnot \macltlfo$. By
  \lemref{lem:NetFirSeqConstrAutAcc}, there is a finite maximal run
  $\macmarks' =\macmark_{0}' \macmark_{1}' \cdots \macmark_{
  \macfslidx}'$ such that the word
  $(\macmark_{0}' \macrestr{\macsplaces})
  (\macmark_{1}'\macrestr{\macsplaces}) \cdots (\macmark_{
  \macfslidx}'\macrestr{\macsplaces})$ is a valid run of
  $\macfsa_{\macnet}$ and for every position $\macseqidx$ of
  $\macmarks'$, $\macmark_{\macseqidx}'
  \macrestr{\macplaces_{\macltlfo}} = \macmark_{\macseqidx}
  \macrestr{\macplaces_{\macltlfo}}$. By \propref{prop:LTLInvOtherPl},
  $(\macmark_{0}' \macrestr{\macplaces_{\macltlfo}})
  (\macmark_{1}'\macrestr{\macplaces_{\macltlfo}}) \cdots (\macmark_{
  \macfslidx}'\macrestr{\macplaces_{\macltlfo}})$ satisfies $\lnot
  \macltlfo$ and hence accepted by $\macfsa_{ \lnot \macltlfo}$, say
  with the run $\macstateo_{0}' \macstateo_{1}' \cdots
  \macstateo_{\macfslidx}' \macstateo_{\macfinss}'$. Let the word
  $\macplvar_{1} \cdots \macplvar_{\macfslidx} \in \macalph^{*}$
  witness the validity of the run $(\macmark_{0}'
  \macrestr{\macsplaces}) (\macmark_{1}'\macrestr{\macsplaces}) \cdots
  (\macmark_{ \macfslidx}'\macrestr{\macsplaces})$ in $\macfsa_{
  \macnet}$, as in \defref{def:ConstrAutAcc}. By
  \defref{def:NetAutProd}, the sequence $(\macmark_{0}'
  \macrestr{\macplaces_{\macltlfo}},\macstateo_{0}') \macplvar_1
  (\macmark_{1}' \macrestr{\macplaces_{\macltlfo}},\macstateo_{1}')
  \cdots \macplvar_{\macfslidx} (\macmark_{\macfslidx}'
  \macrestr{\macplaces_{\macltlfo}},\macstateo_{\macfslidx}')$ is an
  accepting path of $\macfsa_{ \macnet} \times \macfsa_{ \lnot
  \macltlfo}$. The argument for maximal infinite runs is similar.\qed
\end{proof}
\end{fver}

To efficiently check the existence of accepting paths in
$\macfsa_{\macnet} \times \macfsa_{\lnot \macltlfo}$ and
$\macfsa_{\macnet} \times \macfsba_{\lnot \macltlfo}$, it is
convenient to look at them as graphs, possibly with self loops and
parallel edges. Let the set of states be the set of vertices of the
graph and each transition $(\macstateo, \macplvar_{\macseqidx},
\macstateo')$ be an $\macplvar_{\macseqidx}$-labelled edge leaving
$\macstateo$ and entering $\macstateo'$.  If there is a path
$\macpath$ in the graph from $\macstateo$ to $\macstateo'$, the number
of times an edge $\macedge$ occurs in $\macpath$ is denoted by
$\macpath(\macedge)$. If $\macvertext \notin \{\macstateo,
\macstateo'\}$ is some node occurring in $\macpath$, then the number
of edges of $\macpath$ entering $\macvertext$ is equal to the number
of edges of $\macpath$ leaving $\macvertext$. These conditions can be
expressed as integer linear constraints.
\begin{align}
  \label{eq:GraphPathILPConstr}
  \sum_{\macedge \text{ leaves } \macstateo} \macpath(\macedge) -
  \sum_{\macedge \text{ enters } \macstateo} \macpath(\macedge) &= 1 \nonumber\\
  \sum_{\macedge \text{ enters } \macstateo'} \macpath(\macedge) -
  \sum_{\macedge \text{ leaves } \macstateo'} \macpath(\macedge) &= 1 \\
  \macvertext \notin \{\macstateo, \macstateo'\}: \sum_{\macedge
  \text{ enters }\macvertext} \macpath(\macedge) &= \sum_{\macedge
  \text{ leaves }\macvertext} \macpath(\macedge) \nonumber
\end{align}
%In connected graphs, the above conditions are also sufficient for the
%existence of a path. This result is also useful in other problems of
%general Petri nets \cite[section 2.1]{CRIC1990}.
\begin{lemma}[Theorem 2.1, \cite{CRIC1990}]
  \label{lem:GraphPathILPConstr}
  In a directed graph $\macgraph = (\macvertexs, \macedges)$ (possibly
  with self loops and parallel edges), let $\macpath: \macedges
  \to \macNat$ be a function such that the underlying undirected graph
  induced by edges $\macedge$ such that $\macpath(\macedge) > 0$ is
  connected. Then, there is a path from $\macstateo$ to $\macstateo'$
  with each edge $\macedge$ occurring $\macpath(\macedge)$ times iff
  $\macpath$ satisfies the constraints 
  \eqref{eq:GraphPathILPConstr} above.
\end{lemma}
%\begin{proof}
%  By induction on $\sum_{\macedge \in \macedges} \macpath(\macedge)$.
%\end{proof}
If the beginning and the end of a path are same (i.e., if $\macstateo
= \macstateo'$), small modifications of \eqref{eq:GraphPathILPConstr}
and \lemref{lem:GraphPathILPConstr} are required.  Finally we can
prove our desired theorem.

\begin{theorem}
  \label{thm:LTLModChFPT}
  Let $\macnet$ be a $1$-safe net with initial marking $\macmark_{0}$
  and $\macltlfo$ be a MSO formula. Parameterized by the vertex cover
  number of $\macgraph(\macnet)$ and the size of $\macltlfo$, checking
  whether $\macnet$ is a model of $\macltlfo$ is \macfpt{}.
\end{theorem}
\begin{proof}
  By \propref{prop:ModelCheckAccPath}, it is enough to check that
  there is no accepting paths in $\macfsa_{\macnet} \times
  \macfsa_{\lnot \macltlfo}$ and $\macfsa_{\macnet} \times
  \macfsba_{\lnot \macltlfo}$. To check the existence of accepting
  paths in $\macfsa_{\macnet} \times \macfsba_{\lnot \macltlfo}$, we
  have to check if from some initial state in $\{\macmark_{0}
  \macrestr{\macsplaces}\}\times \macstates_{0 \macfsba}$, we can reach
  some vertex in a maximal strongly connected component induced by
  $\bot$-labelled edges, which contains some states from
  $\macstates_{\macnet} \times \macfinss_{\macfsba}$. For every such
  initial state $\macstateo$ and a vertex $\macstateo'$ in such a
  strongly connected component, check the feasibility of
  \eqref{eq:GraphPathILPConstr} along with the following constraint
  for each interface $\macplvar$:
  \begin{align}
    \label{eq:EdgeConstr}
    \sum_{\macedge \text{ is } \macplvar-\text{ labelled}}
    \macpath(\macedge) \le \macedgeconstr(\macplvar)
  \end{align}

  To check the existence of accepting paths in $\macfsa_{\macnet}
  \times \macfsa_{\lnot \macltlfo}$, check the feasibility of
  \eqref{eq:GraphPathILPConstr} and \eqref{eq:EdgeConstr} for every
  state $\macstateo$ in $\{\macmark_{0} \macrestr{\macsplaces}\}
  \times \macstates_{0 \macfsa}$ and every state $(\macplaces_{1},
  \macstateo'')$ in $\macfinss_{\macnet} \times \macstates_{\macfsa}$
  with some $\macstateo_{\macfinss} \in \macfinss_{\macfsa}$ such that
  $(\macstateo'', \macplaces_{1} \cap \macplaces_{\macltlfo},
  \macstateo_{\macfinss}) \in \mactransrel_{\macfsa}$. If some marking
  $\macmark$ with $\macmark\macrestr{\macsplaces} = \macplaces_{1}$
  enables some transition removing a token from some place with
  interface $\macplvar$, then for each such interface, add the
  following constraint:
  \begin{align}
    \label{eq:EdgeConstrMax}
    \sum_{\macedge \text{ is } \macplvar-\text{ labelled}}
    \macpath(\macedge) = \macedgeconstr(\macplvar)
  \end{align}

  The variables in the above \macilp{} instances are
  $\macpath(\macedge)$ for each edge $\macedge$. The number of
  variables in each \macilp{} instance is bounded by some function of
  the parameters. As \macilp{} is \macfpt{} when parameterized by the
  number of variables \cite{Kannan87,Lenstra83,FT87}, the result
  follows.
\qed
\end{proof}

The dependence of the running time of the above algorithm on formula
size is non-elementary if the formula is MSO \cite{ARM75}. The
dependence reduces to single exponential in case of LTL formulas
\cite{Vardi96}. The dependence on vertex cover number is dominated by
the running time of \macilp{}, which is singly exponential in the
number of its variables. The number of variables in turn depends on
the number of VC-interfaces (\defref{def:placeVariety}). In the worst
case, this can be triply exponential but a given $1$-safe Petri net
need not have all possible VC-interfaces.

\section{Conclusion}
The main idea behind the \macfpt{} upper bound for MSO/LTL model
checking is the fact that the problem can be reduced to graph
reachability and hence to \macilp{}. It remains to be seen if such
techniques or others can be applied for branching time logics such as
CTL.

We have some negative results with pathwidth and benefit depth
as parameters and a positive result with vertex cover number as
parameter. We think it is a challenging problem to identify other
parameters associated with $1$-safe Petri nets for which standard
problems in the concurrency literature are \macfpt{}.
Another direction for further work, suggested by a referee,
is to check if the upper bound can be extended to other classes of 
Petri nets such as communication-free nets. 

The results of \secref{sec:twhardness} proves hardness for the lowest
level of the W-hierarchy.  It remains to be seen if the
lower bounds could be made tighter. The parameterized classes
\macparanp{} and \macxp{} include the whole W-hierarchy. Lower bounds
or upper bounds corresponding to these classes would be interesting.

\bibliographystyle{plain}
\bibliography{references}

\end{document}